\documentclass[11pt]{article}
\usepackage{fullpage,amsmath,amsfonts,comment} 
\usepackage{amsthm,amssymb,framed}
\usepackage{times}
\usepackage[colorlinks=true,urlcolor=magenta,linkcolor=blue,citecolor=blue]{hyperref}

\usepackage{cleveref}

\newenvironment{reminder}[1]{\medskip
	\noindent {\bf Reminder of #1.  }\em}{\smallskip}

\newtheorem{theorem}{Theorem}[section]

\newtheorem{lemma}[theorem]{Lemma}
\newtheorem{hypothesis}[theorem]{Hypothesis}
\newtheorem{remark}[theorem]{Remark}

\newtheorem{claim}[theorem]{Claim}

\newtheorem{corollary}[theorem]{Corollary}

\def \poly {{\text{poly}}}

\def \P {{\sf P}}

\def \RP {{\sf RP}}

\def \PSPACE {{\sf PSPACE}}
\def \SPACE {{\sf SPACE}}
\def \N {{\mathbb N}}
\def \F {{\mathbb F}}
\def \eps {{\varepsilon}}

\def \T {{\sf TIME}}
\def \NT {{\sf NTIME}}

\def \NTIME {\NT}
\def \TIME {\T}
\def \NL {{\sf NL}}

\def \BPP {{\sf BPP}}

\def \eps {\varepsilon}
\def \poly {\text{poly}}

\usepackage{etoolbox}
\makeatletter
\patchcmd{\@bibitem}{\doi}{\small\doi}{}{}
\makeatother

\title{Simulating Time With Square-Root Space\footnote{To appear in STOC 2025.}}

\author{Ryan Williams\thanks{Work supported in part by NSF CCF-2127597 and NSF CCF-2420092.}
\\ MIT}

\date{}

\begin{document}

\maketitle

\begin{abstract} We show that for all functions $t(n) \geq n$, every multitape Turing machine running in time $t$ can be simulated in space only $O(\sqrt{t \log t})$. This is a substantial improvement over Hopcroft, Paul, and Valiant's simulation of time $t$ in $O(t/\log t)$ space from 50 years ago [FOCS 1975, JACM 1977]. Among other results, our simulation implies that bounded fan-in circuits of size $s$ can be evaluated on any input in only $\sqrt{s} \cdot \poly(\log s)$ space, and that there are explicit problems solvable in $O(n)$ space which require $n^{2-\eps}$ time on a multitape Turing machine for all $\eps > 0$, thereby making a little progress on the $\P$ versus $\PSPACE$ problem.

Our simulation reduces the problem of simulating time-bounded multitape Turing machines to a series of implicitly-defined Tree Evaluation instances with nice parameters, leveraging the remarkable space-efficient algorithm for Tree Evaluation recently found by Cook and Mertz [STOC 2024].
\end{abstract}

\thispagestyle{empty}
\setcounter{page}{0}

\newpage

\section{Introduction}

One of the most fundamental questions in theoretical computer science is: \emph{how does time relate to space, in computation?} For instance, can every problem solvable in polynomial space ($\PSPACE$) also be solved in polynomial time ($\P$)? Although it is widely believed that $\P \neq \PSPACE$, there has been scant progress on separating time and space over the decades~\cite{DBLP:conf/focs/StearnsHL65,DBLP:journals/jacm/HopcroftU68a,DBLP:conf/stoc/StockmeyerM73,DBLP:journals/jacm/HopcroftPV77,DBLP:journals/jcss/PaulR81,DBLP:journals/mst/HalpernLMW86}. In particular, for general models of computation (beyond one-tape models), only time lower bounds of the form $\Omega(n \log n)$ have been reported for linear-space problems, whereas $\P \neq \PSPACE$ requires showing that there is a linear-space problem that cannot be solved in $O(n^k)$ time for every constant $k \geq 1$. 

In this paper, we make another step towards separating $\P$ from $\PSPACE$, by showing that there are problems solvable in $O(n)$ space that cannot be solved in $n^2/\log^c n$ time for some constant $c > 0$. This is the first generic polynomial separation of time and space in a robust computational model (namely, the multitape Turing machine). The separation is accomplished by exhibiting a surprisingly space-efficient simulation of generic time-bounded algorithms. More formally, let $\TIME[t(n)]$ be the class of decision problems decided by $O(t(n))$-time multitape Turing machines, and let $\SPACE[s(n)]$ the class  decided by $O(s(n))$-space multitape Turing machines.

\begin{theorem}\label{thm:main} For every function $t(n) \geq n$, $\TIME[t(n)] \subseteq \SPACE[\sqrt{t(n) \log t(n)}]$.
\end{theorem}

We find \Cref{thm:main} to be very surprising. It appears to have been a common belief for decades that $t$ time cannot be simulated in $t^{1-\eps}$ space, for any constant $\eps > 0$. For example, assuming that $t$ time cannot be simulated in $t^{1-\eps}$ space, Sipser gave a poly-time derandomization of $\RP$~\cite{DBLP:journals/jcss/Sipser88}\footnote{Sipser also assumed explicit constructions of certain bipartite expander graphs, which were later constructed~\cite{DBLP:journals/jacm/SaksSZ98}.}  and Nisan-Wigderson gave a subexponential-time derandomization of $\BPP$~(\cite[Section 3.6]{DBLP:journals/jcss/NisanW94}).\footnote{Note that the derandomization of Nisan and Wigderson can also be based on assuming the weaker hypothesis $\TIME[t] \not\subset \Sigma_2 \TIME[t^{1-\eps}]$, which remains open and almost certainly cannot be refuted using the ideas of this paper.}

Given the power of multitape Turing machines, \Cref{thm:main} has many interesting consequences. Along with the ability to evaluate generic straight-line programs of length $n$ (over a bounded domain) in only $\tilde{O}(\sqrt{n})$ space (see \Cref{sec:bp}), \Cref{thm:main} immediately implies by diagonalization that there are $s(n)$-space problems that cannot be solved in $s(n)^{2-\eps}$ time for all $\eps > 0$: a ``truly-polynomial'' time lower bound. 

\begin{corollary}\label{cor:space-lb} For space constructible $s(n) \geq n$, and all $\eps > 0$,
$\SPACE[s(n)] \not\subset \TIME[s(n)^{2-\eps}]$.
\end{corollary}

Similarly, if \Cref{thm:main} could be extended to show for \emph{all} $\eps > 0$ that every time-$t(n)$ Turing machine can be simulated in $O(t(n)^{\eps})$ space, then $\P \neq \PSPACE$ would follow (see the arguments in \Cref{sec:consequences}). We discuss this possibility at the end of the paper, in Section~\ref{sec:discussion}. It follows from \Cref{cor:space-lb} that any complete problem for linear space requires quadratic time to be solved. For example:

\begin{corollary}\label{cor:quad-lb} The language $L = \{\langle M,x,1^k\rangle \mid |M| \leq k\text{ and }M(x)\text{ halts in space } k\}$ requires $n^{2-\eps}$ time to solve on a multitape Turing machine, for every $\eps > 0$.    
\end{corollary}

Another consequence is that there are \emph{context-sensitive languages} which need essentially $n^2$ time to be recognized by multitape Turing machines. This follows from the fact that $\mathsf{NSPACE}[n]$  (nondeterministic linear space) is equivalent to the class of context-sensitive languages~\cite{Kuroda64}.

Since multitape Turing machines can evaluate any circuit of size $s$ (and fan-in two) on any given input of length $n \leq s$ in only $s \cdot \poly(\log s)$ time~\cite{Pippenger77}, it follows that arbitrary subquadratic-size circuits can be simulated by subexponential-size branching programs.\footnote{Unfortunately, the reference \cite{Pippenger77} does not seem to be available online. See \Cref{sec:bp} for an alternative sketch of the argument.}

\begin{corollary}\label{cor:bp} There is a universal $k \geq 1$ such that for all $s \geq n$, every bounded fan-in circuit of size $s$ and $n$ inputs has a branching program of size at most $2^{k \sqrt{s} \log^k s}$.
\end{corollary}

It is an interesting question whether the simulation of \Cref{thm:main} can be extended beyond multitape Turing machines, to more powerful computational models. We can generalize \Cref{thm:main} to $d$-dimensional multitape Turing machines, as follows:

\begin{theorem}\label{thm:d-dimensional}
    Every decision problem solvable by a $t(n)$-time $d$-dimensional multitape Turing machine can be decided in $O((t(n) \log t(n))^{1-1/(d+1)})$ space (on a typical one-dimensional multitape Turing machine).
\end{theorem}

The space bound of \Cref{thm:d-dimensional} matches the best known space bounds for simulating time-$t$ machines with \emph{one} $d$-dimensional tape~\cite{DBLP:journals/tcs/Loui81,DBLP:journals/siamcomp/LiskiewiczL90}. See \Cref{sec:related-work} for more references on prior simulations.

\begin{remark}[Extension to Oblivious Random-Access Models] \label{rem:RAM} At the present time, we do not know how to extend \Cref{thm:main} to arbitrary random-access models of computation. The main issue is that the indegree of the resulting computation graph (defined in \Cref{sec:main}) is so high that the computation graph cannot be stored in $O(t^{1-\eps})$ memory. However, we remark that if the pattern of reads and writes of a given random-access machine model is \emph{oblivious}, in the sense that given a step-count $i = 1,\ldots,t$ specified in $O(\log t)$ bits, the names of the registers being accessed in step $i$ can be computed in $O(\sqrt{t})$ space, then \Cref{thm:main} does apply, with only a $\poly(\log t)$ extra space factor. This is because such machines can be simulated by circuits of $t \cdot \poly(\log t)$ size, which can in turn be simulated efficiently by multitape Turing machines in $t \cdot \poly(\log t)$ time (see \Cref{sec:bp}).
\end{remark}

\subsection{Intuition} 

The key idea behind \Cref{thm:main} is to reduce the problem of simulating arbitrary time-bounded computations to particular instances of the {\sc Tree Evaluation} problem, defined by S.~Cook, McKenzie, Wehr, Braverman, and Santhanam~\cite{DBLP:journals/toct/CookMWBS12}. In this problem, one is given a complete $d$-ary tree of height $h$, where each leaf of the tree is labeled with a $b$-bit string, and each inner node of the tree is labeled with a function from $d\cdot b$ bits to $b$ bits.\footnote{Our notation is slightly different from the usual setup, where there is a parameter $k=2^b$. Here, our notation more closely follows Goldreich's exposition~\cite{Goldreich-CM24}.} (Each function is presented as a table of $2^{d \cdot b}$ values, each of which are $b$ bits.) Each inner node computes its value by evaluating its function on the values of its $d$ children. The task of {\sc Tree Evaluation} is to determine the value of the root of the tree. To make this a decision problem, we will decide whether the first bit of the root's value equals $1$ or not. 

The obvious depth-first algorithm for {\sc Tree Evaluation} uses $O(d\cdot b \cdot h)$ space, to store intermediate results at each level of the tree. The authors of~\cite{DBLP:journals/toct/CookMWBS12} conjectured that {\sc Tree Evaluation} is not in $\NL$, which would separate $\NL$ from $\P$ (in fact, stronger separations would hold). Recent algorithmic work has led to significant skepticism that this conjecture is true. In a line of work, J.~Cook and Mertz have found surprisingly space-efficient methods for {\sc Tree Evaluation}~\cite{DBLP:conf/stoc/CookM20,DBLP:journals/eccc/CookM21,DBLP:conf/coco/CookM22,DBLP:conf/stoc/CookM24}, culminating in a marvelous algorithm using space only $O(d \cdot b + h \log (d \cdot b))$~(\cite[Theorem 7]{DBLP:conf/stoc/CookM24}, see \Cref{appendix:cm}).

Utilizing the old notion of ``block-respecting'' Turing machines~\cite{DBLP:journals/jacm/HopcroftPV77}, we show how to reduce time-$t$ computations to (implicitly defined) {\sc Tree Evaluation} instances of height $h = \Theta(t/b)$, bit-length $b$, and fan-in $d$, where $b$ is a parameter from $1$ to $t$ that we can choose, and $d$ is a fixed constant depending on the machine. In particular, one can generate arbitrary bits of our {\sc Tree Evaluation} instance in a space-efficient way as needed. 

Very roughly speaking, each level of the tree will correspond to a ``time block'' of $b$ steps, and the value of each node $v$ of the tree will be a ``block of tape'' of length $\Theta(b)$ (a sequence of $\Theta(b)$ contiguous cells) from some particular tape of a $k$-tape Turing machine; the value of the root will contain the contents of the final tape block of the computation, including the relevant accept/reject state. The evaluation function $F_v$ at node $v$ will simulate the Turing machine for $\Theta(b)$ steps, using the contents of $k$ tape blocks of length $\Theta(b)$ from previous time blocks; this $F_v$ can be computed in $O(b)$ time and space, given the relevant contents of $O(b)$-length tape blocks from all $k$ tapes and the state from the previous time block. (The leaves of the tree roughly correspond to strings of $\Theta(b)$ blank symbols, or a block of $\Theta(b)$ symbols from the input $x$.) We end up with a {\sc Tree Evaluation} instance of height $h=O(t/b)$ and fan-in $d = O(1)$, where $\Theta(b)$ steps of the Turing machine are processed at each node, and where each node is labeled with a string of about $\Theta(b)$ bits.  (There are several technical details to worry over, such as the problem of knowing which tape blocks to use at each time step, but this is the high-level idea.)

Observe that, under the above parameter settings, the obvious depth-first {\sc Tree Evaluation} procedure would yield an algorithm running in $\Theta(h \cdot b) = \Theta(t)$ space. Applying the Cook-Mertz procedure, and setting $d \cdot b  = h \log (d \cdot b) = \Theta(t \cdot \log(d \cdot b)/b)$ to optimize the space usage, we find $b = \Theta(\sqrt{t \log t})$ and obtain the $O(\sqrt{t \log t})$ space bound of \Cref{thm:main}.

\paragraph{Organization.} \Cref{sec:preliminaries} discusses preliminaries (which we do not recommend skipping, but it is short). \Cref{sec:main} begins by proving a short ``warm-up'' simulation (\Cref{thm:warmup}) which already achieves $O(\sqrt{t} \log t)$ space and gives many of the key ideas. \Cref{sec:mainmain} proves \Cref{thm:main}, discusses the possibility of improvement, and proves \Cref{thm:d-dimensional}. \Cref{sec:consequences} discusses various corollaries mentioned earlier, and \Cref{sec:discussion} concludes with a number of new directions to consider.

\section{Preliminaries}\label{sec:preliminaries}

We assume the reader is familiar with basic notions in time and space bounded complexity~\cite{Goldreich08,AB09}. There are a few pieces not covered in the usual complexity textbooks which are important for this work.

\paragraph{Robustness of Space Complexity.} It is important to recall that the class of problems $\SPACE[s(n)]$ is robust under changes in the definition of the underlying machine model. For essentially any model (with sequential access versus random access, or tapes versus registers) that measures space usage by the \emph{total number of bits needed to represent the storage}, $\SPACE[s(n)]$ is the same complexity class. This is part of what is known as the ``invariance thesis''~\cite{DBLP:journals/iandc/SlotB88,DBLP:books/el/leeuwen90/Boas90}. As a consequence, we do not have to worry much about machine models when we're trying to design a space $s(n)$ algorithm and we're ignoring time complexity. This allows us to be more lax in our exposition. 

\paragraph{Block-Respecting Turing Machines.} In our main simulation (\Cref{thm:main}) we will utilize a construction to simplify the analysis of multitape Turing machines. The construction was originally used by Hopcroft, Paul, and Valiant~\cite{DBLP:journals/jacm/HopcroftPV77} in their $O(t/\log t)$-space simulation of $t$ time, and it has been used in many subsequent works to ``decompose'' time-bounded computations (such as \cite{DBLP:conf/focs/PaulPST83,DBLP:journals/tcs/KarakostasLV03,DBLP:journals/tcs/KalyanasundaramLRS12,DBLP:journals/cc/LiptonW13,DBLP:conf/coco/MurrayW17}). A \emph{time-$t(n)$ block-respecting multitape Turing machine with blocks of length $b(n)$} has the property that on all inputs of length $n$, the computation is partitioned into $O(t(n)/b(n))$ \emph{time blocks} of length $b(n)$, while each tape is partitioned into $O(t(n)/b(n))$ contiguous \emph{tape blocks} of length $b(n)$. The key property is that for every time block, each tape head is in \emph{exactly one} tape block during that time block, so that tape heads can switch between tape blocks only at the end of a time block. In particular, every tape head only crosses from one tape block to another, on time steps that are integer multiples of $b(n)$. A key lemma of~\cite{DBLP:journals/jacm/HopcroftPV77} is that every multitape Turing machine can be made block-respecting with low overhead. 

\begin{lemma}[\cite{DBLP:journals/jacm/HopcroftPV77}] \label{lem:block-respect}
    For every time-constructible $b(n),t(n)$ such that $\log t(n) \leq b(n) \leq t(n)$ and every $t(n)$-time $\ell$-tape Turing machine $M$, there is an equivalent $O(t(n))$-time block-respecting $(\ell+1)$-tape Turing machine $M'$ with blocks of length $b(n)$.    
\end{lemma}

\paragraph{The Tree Evaluation Problem.} As mentioned earlier, we will reduce time-$t$ computations to the {\sc Tree Evaluation} problem. Our formulation of {\sc Tree Evaluation} will be relaxed from the original notion: we allow a tree of height \emph{at most} $h$, where each inner node $v$ of the tree has $k_v \leq d$ children for some integer $k_v \geq 2$ depending on $v$, and $v$ is labeled with a function from $k_v \cdot b$ bits to $b$ bits. As before, each leaf is labeled with a $b$-bit string, and we wish to determine the first bit of the root's value. Prior work on {\sc Tree Evaluation} assumed a \emph{complete} $d$-ary tree where all root-to-leaf paths have length equal to $h$; in our setting, we allow some paths to have length than $h$, and some nodes to have fewer than $d$ children. However, the algorithm of Cook and Mertz works just as well in our relaxed formulation:

\begin{theorem}[\cite{DBLP:conf/stoc/CookM24}, Theorem 7] \label{thm:cm} {\sc Tree Evaluation} on trees of bit-length $b$, maximum height $h$, and fan-in \emph{at most} $d$, can be computed in $O(d \cdot b + h \log (d \cdot b))$ space.
\end{theorem}

In Appendix~\ref{appendix:cm}, we give an overview of how the Cook-Mertz algorithm works, and describe why it extends to our case. 

\subsection{More Related Work}

\label{sec:related-work}
As mentioned earlier, Hopcroft, Paul, and Valiant showed that time-$t$ multitape Turing machines can be simulated in $O(t/\log t)$ space~\cite{DBLP:journals/jacm/HopcroftPV77}. This simulation was later extended beyond the multitape model, yielding a more space-efficient simulation for essentially all common models of computation used in algorithms and complexity~\cite{DBLP:journals/jcss/PaulR81,DBLP:journals/mst/HalpernLMW86}. Paterson and Valiant~\cite{DBLP:journals/tcs/PatersonV76} showed that circuits of size $s$ can be simulated by depth $O(s/\log s)$ circuits, implying a space-efficient simulation of circuits. Similarly, Dymond and Tompa~\cite{DBLP:journals/jcss/DymondT85} showed that time $t$ can be simulated in alternating time $O(t/\log t)$. 

For decades, a square-root space simulation of the form in \Cref{thm:main} has been known for \emph{one-tape} Turing machines: Hopcroft and Ullman~\cite{DBLP:journals/jacm/HopcroftU68a} showed that time-$t$ Turing machines with one read-write tape can be simulated in space $O(\sqrt{t})$. Other improvements on this simulation (e.g., improving the time of the simulation, improving the model(s) slightly) include~\cite{DBLP:journals/jcss/Paterson72,DBLP:journals/siamcomp/IbarraM83,DBLP:journals/siamcomp/LiskiewiczL90}. 

Paul, Pippenger, Szemer\'{e}di, and Trotter~\cite{DBLP:conf/focs/PaulPST83} proved the separation $\NTIME[n] \neq \TIME[n]$ for multitape Turing machines. Unfortunately, their proof \emph{only} works for (one-dimensional) multitape Turing machines, and it is infamously still open to prove $\NTIME[n] \neq \TIME[n]$ in more general models. We prove \Cref{thm:d-dimensional} (an extension to the $d$-dimensional case) to illustrate that our approach via {\sc Tree Evaluation} is more broadly applicable.

\section{A More Space-Efficient Simulation of Time}
\label{sec:main}

In this section, we prove \Cref{thm:main}, showing that any multitape Turing machine $M$ running in time $t$ can be simulated in space $O(\sqrt{t \log t})$. We will proceed by reducing the computation of $M$ on an input $x$ (of length $n$) to an instance of {\sc Tree Evaluation} with particular parameters. (In fact, in the full proof of \Cref{thm:main}, we will reduce computing $M$ on $x$ to a series of {\sc Tree Evaluation} instances.)

One initial remark is in order. First, note that we did not specify in the statement of \Cref{thm:main} that the time function $t(n)$ is \emph{time constructible}, in that there is a Turing machine that can print the value of $t(n)$ on $1^n$ in $O(t(n))$ steps. This is because in our space-bounded simulation, we can simply try increasing values $t(n) = n, n+1, n+2, \ldots$, one at a time, and not worry about constructibility issues. (This trick is used in other space-bounded simulations, such as \cite{DBLP:journals/jacm/HopcroftPV77}.) 

Assume $M$ has $\ell$ tapes which are infinite in one direction, and all tape heads are initially at the leftmost cell. We also assume that the input $x$ is received on tape $1$. 

\subsection{A Warm-up Result}

Before describing the full $O(\sqrt{t \log t})$ space simulation, first we present a slightly worse algorithm which uses $O(\sqrt{t} \log t)$ space and requires $t(n) \geq n^2$, but is simpler to describe. This bound is already extremely surprising (to the author), and the simulation gives most of the key intuition as well. 

\begin{theorem}\label{thm:warmup}
 For every function $t(n) \geq n^2$, $\TIME[t(n)] \subseteq \SPACE[\sqrt{t(n)} \log t(n)]$.
\end{theorem}

The proof of \Cref{thm:warmup} will utilize the well-known \emph{oblivious} two-tape simulation of multitape Turing machines. A multitape Turing machine $M$ is \emph{oblivious} if for every $n$, and every input $x$ of length $n$, the tape head movements of $M$ on $x$ depend \emph{only} on $n$, and not on $x$ itself. 

\begin{theorem}[\cite{DBLP:journals/jacm/HennieS66,DBLP:journals/jacm/PippengerF79,DBLP:journals/jacm/FortnowLMV05}]\label{thm:oblivious}
For every time-$t(n)$ multitape Turing machine $M$, there is an equivalent time-$T(n)$ two-tape Turing machine $M'$ which is oblivious, with $T(n) \leq O(t(n) \log t(n))$. Furthermore, given $n$ and $i \in [T(n)]$ specified in $O(\log t(n))$ bits, the two head positions of $M'$ on a length-$n$ input at time step $i$ can be computed in $\poly(\log t(n))$ time.\end{theorem}

Let $M'$ be the machine obtained from \Cref{thm:oblivious}. The first idea behind our simulation is to conceptually partition the computation of $M'$ on a length-$n$ input $x$ into time and tape ``blocks'' of length $b(n)$, for a parameter $b(n) \geq \log t(n)$ to be set later. In particular, the two tapes of $M'$ are split into \emph{tape blocks} of $b(n)$ contiguous cells, and the $T(n)$ steps of $M'$ on $x$ are split into $B := B(n) = O(T(n)/b(n))$ contiguous \emph{time blocks} of length up to $b(n)$. Observe that, for any block of $b(n)$ steps on $M'$, and any given tape $h \in \{1,2\}$, there are at most two tape blocks of tape $h$ that may have been accessed during the time block (moreover, they are adjacent tape blocks on tape $h$). We will construct a {\sc Tree Evaluation} instance where the functions at each node of the tree evaluate a single time block, taking as input some relevant tape blocks from previous time blocks. In fact, the same time block may be recomputed many times over the entire {\sc Tree Evaluation} instance, in order to evaluate the last time block at the root of the tree. 

\paragraph{A Computation Graph.} To this end, we define a \emph{computation graph} $G_{M',x}$ on $B+1 = O(T(n)/b(n))$ nodes, a directed acyclic graph whose edges model the information flow from earlier time blocks in the computation to later time blocks: the reads and writes of $M'$ and the head movements across blocks of tape. Our notion is very similar to the computation graphs defined in \cite{DBLP:journals/jacm/HopcroftPV77,DBLP:journals/acta/PaulR80}. Eventually, the computation being performed on the $(B+1)$-node graph $G_{M',x}$ will be viewed as a {\sc Tree Evaluation} instance of height $B+1$. 

Our graph $G_{M',x}$ has a node $i$ for each time block $i \in \{0,1,\ldots,B\}$, and the edges will indicate which previous time block contents need to be read in order to compute the content of the tape blocks accessed during time block $i$. In particular, we say that all tape blocks on the two tapes are \emph{active} during time block $0$, and for $i > 0$, a tape block is \emph{active} during time block $i$ if the tape head visits some cell of the block during time block $i$. We put an edge from $(i, j)$ in $G_{M',x}$ with $i < j$ if and only if:
\begin{itemize}
\item either $j=i+1$, or
\item when $M'$ is run on input $x$, there is some tape block active during time block $i$ that is not active again until time block $j$. That is, for some tape head $h$, it reads the same tape block $C$ in both time blocks $i$ and $j$, but $h$ does not read tape block $C$ during any time blocks $i+1,\ldots,j-1$. (Alternatively, if some tape block is being accessed for the first time in time block $i$, we have an edge $(0,i)$ to reflect the fact that all tape blocks are defined to be active in time block $0$.)
\end{itemize}
Observe that each node $i$ has indegree at most $5$: one edge from $i-1$, and at most four other edges for (at most) four active tape blocks during time block $i$ (two active tape blocks for each of the two tapes).

The key insight behind the computation graph notion is that the information needed to simulate $M'$ during a time block $j$ only requires knowing the information computed in previous time blocks $i$, where $(i,j)$ is an edge in $G_{M',x}$. In particular, the state and head positions of $M'$ at the start of time block $j$ may be obtained from the state and head positions at the end of time block $j-1$, so we have the edge $(j-1,j)$, and we have an edge $(i,j)$ for each of those blocks of tape accessed during a time block $i$ which are not accessed again until time block $j$. 

Due to the obliviousness of $M'$ (\Cref{thm:oblivious}), we have the following claim:

\begin{claim}\label{warmup-claim}
Given the indices of any two nodes $i$, $j$ in $G_{M',x}$, we can determine if there is an edge from $i$ to $j$ in $\poly(\log t(n))$ additional space.
\end{claim}

The claim follows because determining whether $(i,j)$ is an edge just requires us to keep track of the locations of the two tape heads, from some time block $i$ to a later time block $j$. By \Cref{thm:oblivious}, we can calculate the two tape head locations at any point in time using only $\poly(\log t(n))$ time, so we only have to use $\poly(\log t(n))$ space to keep track of whether a tape block accessed during time block $i$ is only accessed later at time block $j$. 

\paragraph{The Functions at the Nodes.} Now we define what is being computed at each node of $G_{M',x}$. The \emph{content} of time block $i$, denoted by $\text{content}(i)$, is defined as follows:
\begin{itemize}
    \item $\text{content}(0)$ is the initial configuration of $M'$ running on the input $x$ of length $n$, encoded in $n+O(1)$ bits. (Recalling we have $t(n) \geq n^2$, we will eventually set $b(n) \geq n$. Thus the initial configuration of $M'$ on $x$ can ``fit'' in one tape block.)
    
    \item For $i \in \{1,\ldots,B\}$, $\text{content}(i)$ encodes information about the status of $M'$ on $x$ at the \emph{end} of time block $i$: the state of $M'$, the two tape head positions, and a list of the contents of those tape blocks accessed during time block $i$. As there are at most four such tape blocks which may have been accessed during time block $i$, $\text{content}(i)$ can be encoded in $O(b(n) + \log t(n)) \leq O(b(n))$ bits.
\end{itemize}
Note that for every fixed $j \in [B]$, if we are given $\text{content}(i)$ for all edges $(i,j)$ in $G_{M',x}$, then we can compute $\text{content}(j)$ in $O(b(n))$ time and space, by simply simulating $M'$ for $b(n)$ steps on the tape blocks of the relevant $\text{content}(i)$ values.

Our goal is to determine $\text{content}(B)$, the content of the final time block, which will contain either an accept or reject state and determine the answer. 

\paragraph{A Tree Evaluation Instance.} To do this, we construct a {\sc Tree Evaluation} instance in which the root node computes $\text{content}(B)$, its children compute $\text{content}(i)$ for all $i$ such that $(i,B)$ is an edge in $G_{M',x}$, and so on; the leaves of our {\sc Tree Evaluation} instance will compute $\text{content}(0)$, the initial configuration of $M'$ on $x$. This  transformation is analogous to how depth-$d$ \emph{Boolean} circuits of fan-in $F$ can be modeled by formulas of depth-$d$ and size at most $O(F^d)$, by tracing over all possible paths from the output gate of the circuit to an input gate of the circuit; see for example~\cite[Chapter 6]{Jukna}.

More precisely, we define a tree $R_{M',x}$ of height at most $B+1$ and fan-in at most 5, with a root node that will evaluate to $\text{content}(B)$. Each node $v$ of $R_{M',x}$ is labeled by a distinct path from some node $j$ in $G_{M',x}$ to the node $B$ of $G_{M',x}$. Inductively, we define the labels as follows:
\begin{itemize}
\item the root node of $R_{G'}$ is labeled by the empty string $\eps$ (the empty path from $B$ to itself), and
\item for every node $v$ in $R_{M',x}$ labeled by a distinct path $P$ from $j$ to $B$, and for every node $i$ with an edge to $j$ in $G_{M',x}$, the node $v$ has a child $w$ in $R_{M',x}$ labeled by the path which takes the edge $(i,j)$ then the path $P$ to $B$. 
\end{itemize}
Observe that paths $P$ of length $\ell$ from a node $j$ to node $B$ can be encoded in $O(\ell)$ bits, since the indegree of each node is at most $5$. Furthermore, given such a path $P$ encoded in this way, observe we can determine the node $j$ at the start of $P$ in $\poly(\log t)$ space, by making repeated calls to the edge relation (as in \Cref{warmup-claim}).

The desired value to be computed at node $v$ (labeled by a path from $j$ to $B$) is precisely $\text{content}(j)$. For $j=0$, this value is just the initial configuration of $M'$ on $x$, which can be produced immediately in $O(n)$ time and space. For $j > 0$, $\text{content}(j)$ can be computed in $O(b(n))$ time and space given the values of the children of $v$. Since $G_{M',x}$ has at most $B+1$ total nodes, the height of $R_{M',x}$ is at most $B+1$. By induction, the value of the root of $R_{M',x}$ is precisely $\text{content}(B)$. 

While the tree $R_{M',x}$ has $2^{\Theta(B)} \leq 2^{\Theta(T(n)/b(n))}$ nodes, observe that for any given node $v$ of $R_{M',x}$ (labeled by a path to $B$), we can easily compute the labels of the children of $v$ in $\poly(\log t(n))$ space, using \Cref{warmup-claim} to compute the edges of the graph $G_{M',x}$. Thus, it takes only $\poly(\log t(n))$ additional space to determine the children of any given node of $R_{M',x}$, and this space can be immediately reused once these children are determined. So without loss of generality, we may assume we have random access to the tree $R_{M',x}$, its leaves, and its functions at each node. 

Finally, we call the Cook-Mertz {\sc Tree Evaluation} algorithm (\Cref{thm:cm}) in its most general form on $R_{M',x}$. Recall that the space bound of this algorithm, for trees of height at most $h$, fan-in at most $d$, computing $b$-bit values at each node, is 
\[O(d \cdot b + h \log (d \cdot b)).\]
For us, $d = 5$, $b = b(n)$, and $h = O(T(n)/b(n)) \leq O(t(n) \log t(n))/b(n)$. We only use $O(b(n))$ additional time and space for each function call to compute some $\text{content}(j)$, and we can reuse this space for every separate call. Therefore our space bound is optimized up to constant factors, by setting $b(n)$ such that \[b(n)^2 = \Theta(t(n) \log t(n) \cdot \log b(n)),\] so $b(n) = \sqrt{t(n)} \cdot \log t(n)$ suffices. This completes the proof of \Cref{thm:warmup}.

\subsection{The Main Result}
\label{sec:mainmain}

Now we describe how to improve the bound of the space simulation, by avoiding the $t \log t$ blowup of the oblivious simulation. This will establish \Cref{thm:main}. The main difficulty is that without the oblivious guarantee of \Cref{thm:oblivious}, we do not know how to determine the edges of computation graph $G_{M,x}$ in an efficient way. To remedy this, we will use more space: we will enumerate over possible computation graphs $G'$, and introduce a method for checking that $G' = G_{M',x}$ in the functions of our {\sc Tree Evaluation} instance. Our graph enumeration will be performed in a particularly space-efficient way, so that if {\sc Tree Evaluation} turns out to be in logspace, the simulation of this paper will yield an $O(\sqrt{t(n)})$ space bound. 

As before, we start with a multitape $M$ which runs in $t(n)$ time, and let $x$ be an input of length $n$.

First, we make $M$ block-respecting as in \Cref{lem:block-respect} with block-length $b = b(n)$ on inputs of length $n$, for a parameter $b$ to be set later. The new multitape machine $M'$ has $p := \ell+1$ tapes, runs in time $O(t(n))$, and has $B := O(t(n)/b(n))$ time and tape blocks.\footnote{Strictly speaking, we do not have to make $M$ block-respecting, but it does make aspects of the presentation a little cleaner: we do not have to reason about ``active'' tape blocks as we did in the warm-up (\Cref{thm:warmup}). Foreshadowing a bit, we will not use a block-respecting notion in the later extension to $d$-dimensional Turing machines (\Cref{thm:d-dimensional}) and again use ``active'' tape blocks.}

\paragraph{The Computation Graph.} We start by refining the computation graph notion from \Cref{thm:warmup}. Here, our computation graph $G_{M',x}$ is similar but not identical to that of~\cite{DBLP:journals/jacm/HopcroftPV77} and the warm-up \Cref{thm:warmup}. 
Because we will allow for space bounds which are smaller than the input length $n=|x|$, we have to make several modifications to $G_{M',x}$ to fit the specifications of {\sc Tree Evaluation}. We define the set of nodes in $G_{M',x}$ to be
\[S = \{(h,i), (h,0,i) \mid h \in [p], i \in [B]\}.\] Intuitively, each $(h,i) \in [p] \times [B]$ will correspond to the content of the relevant block of tape $h$ after time block $i$, while each $(h,0,i)$ will be a source node in $G_{M',x}$ corresponding to the content of the $i$-th block of tape $h$ when it is accessed for the first time, \emph{i.e.}, the \emph{initial configuration} of the $i$-th block of tape $h$.\footnote{A little explanation may be in order. In the warm-up \Cref{thm:warmup}, we assumed $t(n) \geq n^2$, so that the entire initial configuration of the machine on $x$ could fit in a length-$b(n)$ block. This made the leaf nodes of our {\sc Tree Evaluation} instance $R_{G'}$ particularly easy to describe. For $t(n) \ll n^2$, we may have $|x|=n \ll b(n)$, so the input $x$ may not fit in one tape block. To accommodate this possibility and obtain a clean reduction to {\sc Tree Evaluation} in the end, we define multiple source nodes in $G_{M',x}$ to account for different $b(n)$-length blocks of input $x$ in the initial configuration of $M'$ on $x$, along with $b(n)$-length blocks of all-blank tape when these blocks are accessed for the first time.}
(We imagine that on each tape, the tape blocks are indexed $1,2,\ldots$ starting from the leftmost block, with up to $B$ tape blocks for each tape.) We think of all $(h,0,i)$ nodes as associated with ``time block $0$''.

Each node $(h,i) \in [p] \times [B]$ is labeled with an integer $m_{(h,i)} \in \{-1,0,1\}$, indicating the tape head $h$ movement at the end of time block $i$:
\[m_{(h,i)} = 
\begin{cases}
~1 & \text{ if the head $h$ moves one tape block to the right of the current tape block}, \\
-1 & \text{ if $h$ moves one tape block to the left, and} \\
~0 & \text{ if $h$ stays in the same tape block for both time blocks $i$ and $i+1$.} 
\end{cases}
\]

Next, we describe the edges of $G_{M',x}$; there are two types. For each $h, h' \in [p]$ and $i,j \in [B]$ with $i < j$, the edge $((h',i),(h,j))$ is put in $G_{M',x}$ if either:
\begin{itemize}
    \item 
    $j=i+1$, or
    \item while $M'$ is running during time block $j$, the tape block accessed by $h'$ during time block $i$ is not accessed again until time block $j$. That is, tape head $h'$ reads the same tape block $T$ in both time blocks $i$ and $j$, but head $h'$ does not read tape block $T$ during any of the time blocks $i+1,\ldots,j-1$. 
\end{itemize}
For $h,h' \in [p]$ and $i,j \in [B]$, the edge $((h',0,i),(h,j))$ is put in $G_{M',x}$ if, while $M'$ is running during time block $j$, the tape head $h'$ accesses its $i$-th tape block for the first time in the computation. (For example, note that , for all $h,h'$, $((h',0,1),(h,1))$ is an edge.)

Observe that the indegree of each node $(h,j) \in [p] \times [B]$ is at most $2p$ for all $j > 0$: for all $h' \in [p]$, there is an edge from $(h',j-1)$ to $(h,j)$ for $j > 1$ (and from $(h',0,1)$ to $(h,1)$ for $j=1$), and there is an edge from a node labeled by $h'$ (either of the form $(h',i_{h',j})$ or $(h',0,i_{h',j})$) to $(h,j)$, indicating which block of tape $h'$ is needed to compute the block of tape $h$ during time block $j$. (Note that some of these edges might be double-counted, so the indegree is at most $2p$.)

\paragraph{A Succinct Graph Encoding.} The obvious way to store the computation graph $G_{M',x}$ uses $O(B \log B)$ bits. To save space, we will be more careful, and use additional observations on the structure of such computation graphs. (These observations are similar but not identical to those used in the separation of $\NTIME[n]$ and $\TIME[n]$, of Paul-Pippenger-Szemer\'{e}di-Trotter~\cite{DBLP:conf/focs/PaulPST83}.) 

Recall that each node $(h,i)$ is labeled by a head movement $m_{(h,i)} \in \{-1,1,0\}$ indicating whether the tape block of head $h$ is decremented by $1$, incremented by $1$, or stays the same, at the end of time block $i$. Our key observation is that the numbers $m_{(h,i)}$ alone already tell us the entire structure of the graph $G_{M',x}$. This immediately implies that every possible \emph{guess} of $G_{M',x}$ can be encoded in only $O(B)$ bits: we only need a constant number of bits for each of the $O(B)$ nodes.

In particular, for an index $i \in [B]$, we define $\text{block}(h,i) \in [B]$ to be the index of the tape block of tape $h$ being accessed at the \emph{start} of time block $i$. For all $h \in [p]$, we have $\text{block}(h,1) = 1$, and for $i > 1$, by definition of the $m_{(h,i)}$ we have 
\begin{align}\label{eqn:block_from_m}
\text{block}(h,i) = 1 + \sum_{j=1}^{i-1} m_{(h,j)}.
\end{align} 
Equation \eqref{eqn:block_from_m} is true because each $m_{(h,j)}$ tells us how the index of the tape block of tape $h$ changes at the \emph{end} of time block $j$, which is the same as the tape block index at the \emph{start} of time block $j+1$.

For $i < j$, observe that there is an edge from $(h',i)$ to $(h,j)$ in $G_{M',x}$ if and only if either:
\begin{itemize}
    \item $j=i+1$, or 
    \item $\text{block}(h',i) = \text{block}(h',j)$ and for all $k \in \{i+1,\ldots,j-1\}$, $\text{block}(h',k) \neq \text{block}(h',i)$. (The tape block accessed by $h'$ in time block $i$ is not accessed again until time block $j$.)
\end{itemize}
Furthermore, there is an edge from $(h',0,i)$ to $(h,j)$ in $G_{M',x}$ if and only if $i=\text{block}(h',j)$ and for all $k \in \{1,\ldots,j-1\}$, $\text{block}(h',j) \neq \text{block}(h',k)$. (The tape block accessed by $h'$ in time block $j$ was never accessed before, and its index is equal to $i$.)

We will use a few claims about the $\text{block}$ function and the computation graph.

\begin{claim} \label{claim2} Given $(h',i') \in [p] \times [B]$ and the claimed sequence $\{m_{(h,i)}\}$, we can compute $\text{block}(h',i')$ in $O(\log t(n))$ additional space.
\end{claim}

\begin{proof}
    Given any $h',i'$, each $\text{block}(h',i')$ can be computed in logspace using \eqref{eqn:block_from_m}, by maintaining a counter and streaming over the sequence $m_{(h,i)}$.
\end{proof}

\begin{claim} \label{claim} Given the indices of a pair of nodes $u,v$ in $G_{M',x}$, and the claimed sequence $\{m_{(h,i)}\}$, we can determine if $(u,v)$ is an edge in the encoded graph in $O(\log t(n))$ additional space. 
\end{claim}

\begin{proof} We can easily check if $j=i+1$ in logspace. By \Cref{claim2}, we can compute $\text{block}(h',j)$ for any $h',j$ in logspace as well. Therefore the three possible conditions for an edge in the computation graph $G_{M',x}$ can each be checked in logspace.
\end{proof}

To summarize, we can encode $G_{M',x}$ in $O(B)$ bits, and given such an encoding, we can determine any desired edge $(u,v)$ of the graph in logspace.

\paragraph{Values of Nodes.} Similarly to our warm-up \Cref{thm:warmup}, we define contents for each of the nodes of $G_{M',x}$. For all $h \in [p]$ and $i \in [B]$, we define $\text{content}(h,0,i)$ for the source nodes $(h,0,i)$ as follows:
\begin{itemize}
\item If $h > 1$, then we are reading a tape $h$ that does not contain the input. In this case, $\text{content}(h,0,i)$ is defined to be \emph{all-blank tape content}: $b(n)$ blanks, with the tape head at the leftmost cell of the block, and head position equal to $(i-1) \cdot b(n)$. (Note that, assuming we start numbering tape cells at $0$, $(i-1) \cdot b(n)$ is the leftmost cell of the $i$-th block of tape.)
\item If $h=1$, then we may need to read portions of the input $x$ of length $n$. If $i > \lceil n/b(n) \rceil$, then the $i$-th tape block of tape $1$ does not contain any symbol of $x$, and $\text{content}(1,0,i)$ is defined to be all-blank tape content as above. Otherwise, if $i \leq \lceil n/b(n) \rceil$, then $\text{content}(1,0,i)$ is defined to be the relevant $b(n)$-length substring of $x$ (possibly padded with blanks at the end) with the tape head at the leftmost cell of the block, head position equal to $(i-1) \cdot b(n)$, and the initial state of $M'$ included if $i=0$.
\end{itemize}

Note that the source nodes of our $G_{M',x}$ will eventually be the leaf nodes of the {\sc Tree Evaluation} instance; in the above, we are also defining the values of those leaves.

For $h \in [p]$ and $i \in [B]$, we define $\text{content}(h,i)$ of node $(h,i)$ similarly as in \Cref{thm:warmup}: it is a string encoding the pair $(h,i)$, the state of $M'$ and the head position of tape $h$ at the end of time block $i$, and the content of the tape block read by head $h$ at the end of time block $i$. With a reasonable encoding, $\text{content}(h,i)$ can be represented in $O(b(n) + \log t(n)) \leq O(b(n))$ bits, and our computation graph $G_{M',x}$ has been set up (just as in \Cref{thm:warmup}) so that given the strings $\text{content}(u)$ for all nodes $u$ with edges to $(h,j)$, the string $\text{content}(h,j)$ can be computed in $O(b(n))$ time and space.

\paragraph{Computation Graph Enumeration.} In what follows, we enumerate over all possible $O(t(n)/b(n))$-bit choices $G'$ of the encoding of the computation graph $G_{M',x}$ on $O(B) \leq O(t(n)/b(n))$ nodes. For each $G'$, we construct a {\sc Tree Evaluation} instance $R_{G'}$ based on $G'$, which will attempt to simulate $M'$ on $x$ assuming $G' = G_{M',x}$. We will set up $R_{G'}$ so that the following conditions hold:
\begin{itemize}
    \item If $G' \neq G_{M',x}$, this means that at the end of some time block $i$, some tape head $h$ of $M'$ on $x$ moves inconsistently with the guessed label $m_{(h,i)}$ in $G'$. In this case, evaluating $R_{G'}$ will result in a special {\tt FAIL} value at the root. 
    \item If $G' = G_{M',x}$, then $R_{G'}$ will evaluate to $\text{content}(1,B)$ at the root, which will include either an accept or reject state at the end of the computation of $M'$ on $x$. Thus we will be able to conclude decisively whether $M$ accepts or rejects $x$ after this call to {\sc Tree Evaluation}.
    \end{itemize}

Finally, if we exhaust all computation graphs $G'$ and always receive {\tt FAIL} from all {\sc Tree Evaluation} calls, we then increase our guess of $t(n)$ by $1$ (starting from $t(n)=n$), and restart the entire process. (Recall that we do not assume $t(n)$ is constructible.) Eventually, we will choose an appropriate $t(n)$, in which case the above enumeration of graphs $G'$ will result in either acceptance or rejection.

\paragraph{The Functions At The Nodes.} We now define the functions at the nodes of our {\sc Tree Evaluation} instance $R_{G'}$. These functions will allow us to detect when the current graph $G'$ we are considering has an error. We have to be a little careful here, as the Cook-Mertz algorithm is only guaranteed to produce the value of the \emph{root} of a given tree, and not the values of any intermediate nodes along the way. (Indeed, the values of other nodes may become quite ``scrambled'' from evaluating a low-degree extension of the functions involved.)

For each tape $h \in [p]$ and time block $i \in [B]$, we define a \emph{time block function} $F_{h,i}$ which attempts to simulate $M'$ over time block $i$, and to output $\text{content}(h,i)$, the content of the relevant block of tape $h$ at the end of time block $i$. 

First, we choose an encoding of $\text{content}(h,i)$ strings so that there is also a special {\tt FAIL} string of length $O(b(n))$, which is distinct from all valid content strings.
\begin{itemize}
\item The {\bf input} to $F_{h,i}$ consists of up to $2p$ strings of length $O(b(n))$. Some strings may be $O(b(n))$-bit {\tt FAIL} strings. In the case where $G' = G_{M',x}$, $p$ of the input strings have the form $\text{content}(h',i-1)$ (or $\text{content}(h',0,1)$ if $i=1$) for all $h' \in [p]$. These content strings contain the index of the node in $G'$ they correspond to, the state $q$ of $M'$ at the start of time block $i$, the content of the relevant tape blocks at the end of time block $i-1$, and the head positions of all $p$ tapes at the start of time block $i$, where each head position is encoded in $O(\log b(n))$ bits. When some of the tape blocks accessed in time block $i$ were not accessed in the previous time block, there are also (up to) $p$ other strings $\text{content}(u_1),\ldots,\text{content}(u_p)$ which contain tape block content $c_1,\ldots,c_p$ for each of the $p$ tapes at the start of time block $i$. 

\item The {\bf output} of $F_{h,i}$ is defined as follows. First of all, if some input string is a {\tt FAIL} string, then  $F_{h,i}$ immediately outputs {\tt FAIL} as well. In this way, a single {\tt FAIL} detection at any node will propagate to the root value of $R_{G'}$.

Assuming no {\tt FAIL} strings have been given as input, $F_{h,i}$ attempts to simulate $M'$ for time block $i$, using the given content strings. While simulating, $F_{h,i}$ checks that for all $h' \in [p]$, the tape head $h'$ moves consistently with the integer $m_{(h',i)} \in \{0,1,-1\}$. In particular, if $m_{(h',i)} = -1$ then it checks tape head $h'$ moves to its left-adjacent tape block at the end of time block $i$, if $m_{(h',i)} = 1$ then it checks $h'$ moves to its right-adjacent tape block, and if $m_{(h',i)} = 0$ then it checks $h'$ remains in the same tape block. If all heads $h'$ move consistently with the integers $m_{(h',i)}$, then the output of $F_{h,i}$ is set to be $\text{content}(h,i)$. Otherwise, the output of $F_{h,i}$ is {\tt FAIL}.

\end{itemize}
Observe that $F_{h,i}$ can be computed in $O(b(n))$ time and space, by simply simulating $M'$ for $b(n)$ steps, starting from the contents $c_1,\ldots,c_p$ for each of the tapes, and the state and head information given. Furthermore, by setting $b'(n) = \Theta(b(n))$ appropriately, we may think of each $F_{h,i}$ as a function from an ordered collection of $2p$ separate $b'(n)$-bit strings, to a single $b'(n)$-bit string. These will be the functions at the nodes of our {\sc Tree Evaluation} instance.

\paragraph{Construct a Tree Evaluation Instance.} Observe that the depth of every $G'$ is at most $B+1$: for every edge $(u,v)$ in $G'$, the time block of $u$ is always smaller than the time block of $v$, and there are $B+1$ time blocks. For each possible computation graph $G'$, we will construct an equivalent and implicitly-defined {\sc Tree Evaluation} instance $R_{G'}$ of height at most $B+1 \leq O(t(n)/b(n))$, such that the edge relation of the tree can be determined in small space. The idea is analogous to that described in the warm-up (\Cref{thm:warmup}); for completeness, we will be a little more formal than \Cref{thm:warmup}. 

Recall that each candidate $G'$ is defined so that each non-source node $(h,i)$ has indegree at most $2p$. Let $V$ be the set of all $O((2p)^{B+1})$ sequences of the form $h_1 \cdots h_{\ell}$ for all $\ell = 0,\ldots,B$, where each $h_i \in [2p]$, and let $\eps$ denote the empty string (also in $V$). The nodes of $R_{G'}$ will be indexed by sequences in $V$. For all $\ell = 0,\ldots,B-1$, each node $h_1 \cdots h_{\ell} \in V$ has at most $2p$ children $h_1 \cdots h_{\ell} h_{\ell+1}$ for some $h_{\ell+1} \in [2p]$.

Each node of $R_{G'}$ is directly associated with a path from a node $v$ to the node $(1,B)$ in the guessed graph $G'$, as follows. 
\begin{itemize}
\item The node $\eps \in V$ corresponds to the node $(1,B)$ in $G'$ (tape $1$, in the last time block), and we associate the function $F_{1,B}$ with this node. 

\item Inductively assume the node $h_1 \cdots h_{\ell} \in V$ corresponds to some node $(h,j)$ in $G'$. We associate the function $F_{h,j}$ with this node. 

For every $h' \in \{1,\ldots,p\}$, the node $h_1 \cdots h_{\ell} h' \in V$ corresponds to a node $(h',i)$ (or $(h',0,i)$) in $G'$ with an edge to $(h,j)$, for some $i$. This corresponds to the case where the tape block of tape $h'$ accessed in time block $i$ is accessed later in time block $j$ (or when $h'$ is reading tape block $i$ for the first time in time block $j$).

For every $h' \in \{p+1,\ldots,2p\}$, let $h'' = h'-p$, so $h'' \in [p]$. The node $h_1 \cdots h_{\ell} h' \in V$ corresponds to the node $(h'',j-1)$ with an edge to $(h,j)$ when $j>1$, and the node $(h'',0,1)$ with an edge to $(h,j)$ when $j=1$. This corresponds to the case where the state and head position of tape $h''$ at the end of time block $j-1$ is passed to the start of time block $j$, and to the case where the initial state and head position is passed to the start of time block $1$. 

If there is an edge from $(h',0,i)$ to $(h,j)$ for some $i$, then the tape head $h'$ has never previously visited the tape block $i$ that is used to compute time block $j$. In that case, we set $h_1 \cdots h_{\ell} h'$ to be a \emph{leaf node} in the tree. The value of that leaf node is set to be $\text{content}(h',0,i)$.

\end{itemize}

\paragraph{Finishing up.} We call {\sc Tree Evaluation} on $R_{G'}$. If the current guessed $G'$ is not equal to $G_{M',x}$, then some guessed integer $m_{(h,i)} \in \{-1,1,0\}$ is an incorrect value (recall that $G'$ is specified entirely by the integers $\{m_{(h,i)}\}$). This incorrect head movement will be detected by the function $F_{h,i}$, which will then output {\tt FAIL}. This {\tt FAIL} value will propagate to the root value of $R_{G'}$ by construction. When $R_{G'}$ evaluates to {\tt FAIL}, we move to the next possible $G'$, encoded in $O(B)$ bits.

Assuming the current graph $G'$ is correct, \emph{i.e.}, $G' = G_{M',x}$, then the value of the root of $R_{G'}$ is $\text{content}(1,B)$, the output of $F_{1,B}$ on the final time block $B$. This value contains the correct accept/reject state of $M'$ on $x$. This follows from the construction of the functions $F_{h,i}$, and can be proved formally by an induction on the nodes of $G'$ in topological order. Therefore if our {\sc Tree Evaluation} call returns some content with an accept/reject state, we can immediately return the decision. 

\paragraph{Space Complexity.} Observe that, by \Cref{claim2} and \Cref{claim}, any bits of the {\sc Tree Evaluation} instance $R_{G'}$ defined above can be computed in $O(B)$ space, given the computation graph $G'$ encoded in $O(B)$ space. In particular, given the index of a node of $R_{G'}$ of the tree as defined above (as a path from a node $v$ to $(1,B)$ in $G'$), we can determine the corresponding node of $G'$ and its children in $O(B)$ space.

Therefore, we can call the Cook-Mertz algorithm (\Cref{thm:cm}) on this implicitly-defined instance of {\sc Tree Evaluation}, using $O(B)$ space plus $O(d' \cdot b' + h' \cdot \log(d \cdot b'))$ space, where
\begin{itemize}
\item $b' = \Theta(b(n))$ is the bit-length of the output of our functions $F_{h,i}$, 
\item $d' = 2p = \Theta(1)$, the number of children of each inner node, and
\item $h' = B = \Theta(t(n)/b(n))$, the maximum height of the tree.
\end{itemize}
At each node, each call to one of the functions $F_{h,i}$ requires only $O(b(n))$ space. Therefore the overall space bound is \[O\left(b(n) + \frac{t(n)}{b(n)}\cdot \log(b(n))\right).\] Setting $b(n) = \sqrt{t(n) \log t(n)}$ obtains the bound $O(\sqrt{t(n) \log t(n)})$. This completes the proof of \Cref{thm:main}.

Given that we have reduced to {\sc Tree Evaluation}, it may be instructive to think about what is happening in the final algorithm, conceptually. At a high level, for our instances of {\sc Tree Evaluation}, the Cook-Mertz procedure on $R_{G'}$ uses $O((t(n)/b(n)) \cdot \log b(n))$ space to specify the current node $v$ of $G'$ being examined (given by a path from that node $v$ to the node $(1,B)$) as well as an element from a field of size $\poly(b(n))$ for each node along that path. The algorithm also reuses $O(b(n))$ space at each node, in order to compute low-degree extensions of the function $F_{h,i}$ by interpolation over carefully chosen elements of the field. For our setting of $b(n)$, we are using equal amounts of space to store a path in the graph $G'$ labeled with field elements, and to compute a low-degree extension of the ``block evaluation'' function $F_{h,i}$ at each node (with aggressive reuse of the latter space, at every level of the tree).

\subsection{On Possibly Removing the Square-Root-Log Factor}
\label{sec:log}

We observe that, if {\sc Tree Evaluation} turns out to be in ${\sf LOGSPACE} = \SPACE[\log n]$, then we would obtain a simulation that runs in $O(b' + t(n)/b')$ space when the number of children is upper bounded by a constant. This is due to our succinct encoding of the computation graph, which only needs $O(b')$ space to be stored. Setting $b' = \sqrt{t(n)}$, this would remove the pesky $O(\sqrt{\log t})$ factor from our space bound above:

\begin{corollary}
    If {\sc Tree Evaluation} is in ${\sf LOGSPACE}$, then $\TIME[t(n)] \subseteq \SPACE[\sqrt{t(n)}]$.
\end{corollary}

There may be a path to removing the $\sqrt{\log t}$ factor that does not require solving {\sc Tree Evaluation} in logspace. The simulations of time $t$ building on Hopcroft-Paul-Valiant~\cite{DBLP:journals/jacm/HopcroftPV77,DBLP:journals/tcs/PatersonV76,DBLP:journals/jcss/PaulR81,DBLP:journals/jcss/DymondT85,DBLP:journals/mst/HalpernLMW86}, which save a $\log t$ factor in space and alternating time, utilize strategies which seem orthogonal to our approach via {\sc Tree Evaluation}. Roughly speaking, there are two kinds of strategies: a pebbling strategy on computation graphs which minimizes the total number of pebbles needed~\cite{DBLP:journals/jacm/HopcroftPV77,DBLP:journals/jcss/PaulR81}, and a recursive composition strategy which cuts the graph into halves and performs one of two recursive approaches based on the cardinality of the cut~\cite{DBLP:journals/tcs/PatersonV76,DBLP:journals/jcss/DymondT85,DBLP:journals/mst/HalpernLMW86}.\footnote{There is even a third strategy, based on an ``overlap'' argument~\cite{DBLP:journals/mst/AdlemanL81}.} Neither of these seem comparable to the Cook-Mertz algorithm. However, so far we have been unable to merge the {\sc Tree Evaluation} approach and the other strategies. The following hypothesis seems plausible:
\begin{hypothesis} For ``reasonable'' $b(n)$ and $t(n)$, time-$t(n)$ multitape computations can be space-efficiently reduced to {\sc Tree Evaluation} instances with constant degree $d$, height $O(t(n)/b(n))/\log(t(n)/b(n))$, and functions from $O(b)$-bits to $b$-bits at each inner node.
\end{hypothesis}
Results such as \cite{DBLP:journals/tcs/PatersonV76,DBLP:journals/jcss/DymondT85} which show that $\TIME[t] \subseteq {\sf ATIME}[t/\log t]$ and that circuits of size $t$ can be simulated in depth $O(t/\log t)$, prove that the hypothesis is actually true for $b = \Theta(1)$. If the hypothesis is also true for $b = \sqrt{t}$, then we could conclude $\TIME[t] \subseteq \SPACE[\sqrt{t}]$, by applying Cook and Mertz (\Cref{thm:cm}). 

\subsection{Extension to Higher Dimensional Tapes}

The simulation of \Cref{thm:main} extends to Turing machines with higher-dimensional tapes. The proof is very similar in spirit to \Cref{thm:main}, but a few crucial changes are needed to carry out the generalization.

\begin{reminder}{\Cref{thm:d-dimensional}}
    Every decision problem solvable by a $t(n)$-time $d$-dimensional multitape Turing machine can be decided in $O((t \log t)^{1-1/(d+1)})$ space.    
\end{reminder}

\begin{proof} (Sketch) As there is no $d$-dimensional version of block-respecting Turing machines, we have to take a  different approach. Similarly to \Cref{thm:warmup} and Paul-Reischuk~\cite{DBLP:journals/acta/PaulR80}, we upper-bound the number of tape blocks that may be relevant to any given time block, and use that information to construct a series of {\sc Tree Evaluation} instances.

Suppose our time-$t(n)$ Turing machine $M$ has $p$ tapes which are $d$-dimensional, and we wish to simulate $M$ on an input $x$ of length $n$. We assume $x$ is written on the first tape in a single direction starting from the cell indexed by $(0,\ldots,0) \in {\mathbb N}^d$. (For concreteness, for $i=1,\ldots,n$ we may assume the $i$-th bit of $x$ is written in the cell $(0,\ldots,0,i-1) \in {\mathbb N}^d$.)

For a parameter $c$, we partition the time $t(n)$ into time blocks of length $c$, so there are $B = \lceil t(n)/c \rceil$ total time blocks. Besides time blocks $1,\ldots,B$, we also define a time block $0$ (similarly to \Cref{thm:warmup} and \Cref{thm:main}). Each $d$-dimensional tape is partitioned into tape blocks of $b = c^d$ contiguous cells: in particular, each tape block is a $d$-dimensional cube with side-length $c$ in each direction. Observe that each tape block has up to $3^d-1$ adjacent tape blocks in $d$ dimensions. We define a tape block $T$ to be \emph{active} in time block $i > 1$ if some cells of $T$ are accessed during time block $i$, and all tape blocks are defined to be active in time block $0$. Since each time block is only for $c$ steps and each tape block has side-length $c$ in each direction, observe that for each $i \in [B]$ and each tape $h \in [p]$, there are at most $2^d$ blocks of tape $h$ that may be active during time block $i$. Therefore across all $p$ tapes, the total number of active blocks is at most $2^d \cdot p$. Note that each active tape block of tape $h$ during time block $i$ can be indexed by some vector in $\{-1,0,1\}^d$ indicating its position relative to the tape block in which tape $h$ started the time block.

Similar to the proofs of \Cref{thm:warmup} and \Cref{thm:main}, we define a computation graph $G_{M,x}$. The set of nodes will be the set $S = \{(h,i) \mid h \in [p], i \in [B]\}$ unioned with a subset  $T \subseteq \{(h,0,v) \mid h \in [p], v \in \N^d\}$, where $|T|\leq p \cdot B$. 

Each node will have a content value as before. The $\text{content}(h,0,v)$ nodes will store a portion of the input, or the all-blank content, depending on the vector $v \in \N^d$. The vector $v$ gives the \emph{index} of a block of tape $h$. (In the one-dimensional case, the tape blocks were simply indexed by $\N$.) The $\text{content}(h,i)$ nodes will store information at the \emph{end} of time block $i$: the state of $M'$, the head position of tape $h$, and a list of the contents of those blocks of tape $h$ that are active during time block $i$. Note that 
$\text{content}(h,i)$ can be encoded in $O(2^d \cdot b(n))$ bits. 

As before, we label each node $(h,i) \in [p] \times [B]$ to indicate the head movements between time blocks, but our labels are more complicated than before. We use a $d$-dimensional vector $m_{(h,i)} \in \{-1,0,1\}^d$ to indicate the \emph{difference} between the index of the tape block $u_{(h,i)} \in \N^d$ that tape head $h$ is in at the start of time block $i$, and the index of the tape block $v_{(h,i)} \in \N^d$ that head $h$ is in at the end of time block $i$. We have $v_{(h,i)} - u_{(h,i)} \in \{-1,0,1\}^d$, since every time block takes $c$ steps and the side-length of a tape block is $c$ cells. We also label each node with a list $L_{(h,i)}$ of up to $2^d$ other vectors in $\{-1,0,1\}^d$, describing the indices of all other blocks of tape $h$ that are active in time block $i$. Observe that the vectors $m_{(h,i)}$ and the lists $L_{(h,i)}$ over all nodes can be encoded in $O(d \cdot 2^d \cdot B) \leq O(B)$ bits. Moreover, given the vectors $m_{(h,i)}$ and the lists $L_{(h,i)}$, we can reconstruct any $v_{(h,i)}$ and $u_{(h,i)}$ in only logspace by maintaining counters.

As before, we put an edge from $(h,i)$ to $(h',j)$ if either 
\begin{itemize}
\item $j=i+1$, or
\item $i < j$ and there is some tape block active on tape $h$ in both time blocks $i$ and $j$ that is not active for all time blocks $i+1$ through $j-1$.
\end{itemize}
We put an edge from $(h,0,v)$ to $(h',j)$ if during time block $j$, the tape head $h$ accesses the tape block indexed by $v$ for the first time in the computation.

We observe that the indegree of each $(h,i)$ is at most $(2^d+1) \cdot p$: there are at most $2^d$ active tape blocks for each of the $p$ tapes, each of which may require information from a different previous node, and there are also $p$ edges from the previous time block providing the state and head positions at the start of time block $i$. 

Generalizing \Cref{claim2} and \Cref{claim}, we claim that the edges of $G_{M,x}$ can be determined in logspace, given the vectors $m_{(h,i)}$ and the lists $L_{(h,i)}$ encoded in $O(B)$ bits. We enumerate over all possible computation graphs $G'$, using this encoding. Note that the encoding also allows us to determine which source nodes $(h,0,v)$ appear in $G'$, by tracking the tape head movements as claimed by the vectors $m_{(h,i)}$ and the lists $L_{(h,i)}$, and noting when a tape head $h$ enters a new block that has not been accessed before.

For each node $(h,i)$, the evaluation function $F_{h,i}$ is defined similarly as in the proofs of \Cref{thm:warmup} and \Cref{thm:main}. Given all the necessary information at the start of a time block: the state $q$, the contents of the (at most) $2^d \cdot p$ active blocks, and each of the $p$ head positions encoded in $O(d \log t)$ bits, the function $F_{h,i}$ computes $\text{content}(h,i)$: it simulates $M$ for $c$ steps on the active blocks, then outputs updated contents of all $O(2^d)$ active blocks on tape $h$, in $b' \leq O(2^d \cdot c^d + 2^d \cdot d \log t)$ bits for some $b' = \Theta(c^d)$, including the head position of tape $h$ at the end of the time block and the state $q'$ reached.  

As before, $F_{h,i}$ checks that the claimed active tape blocks in $L_{(h,i)}$ and the claimed vector $m_{(h,i)}$ are consistent with the simulation of time block $i$; if they are not, then $F_{h,i}$ outputs {\tt FAIL} and we design the $F_{h,i}$ as before so that a single {\tt FAIL} value is propagated to the root of $R_{G'}$. We can think of each $F_{h,i}$ as a mapping from $(2^d+1) \cdot p \cdot b'$ bits to $b'$ bits, where we allow a special encoding to ``pad'' the input and output when the number of active blocks on some tape is less than $2^d$. 

Finally, for each guessed graph $G'$ we define a {\sc Tree Evaluation} $R_{G'}$ instance analogously as in the proof of \Cref{thm:main}. The root of the tree corresponds to the node $(1,B)$, where we wish to check if $\text{content}(1,B)$ contains the accept or reject state. The children of a given node in the tree $R_{G'}$ are determined directly by the predecessors of the corresponding node in $G_{M,x}$. The leaves correspond to the nodes $(h,0,v)$ in $G_{M,x}$ such that $\text{content}(h,0,v)$ contains either an all-blank $d$-dimensional cube of $c^d$ cells, or a $d$-dimensional cube of $c^d$ cells which includes up to $c$ symbols of the input $x$ and is otherwise all-blank.

As before, if a call to {\sc Tree Evaluation} for some $R_{G'}$ ever returns an answer other than {\tt FAIL}, we return the appropriate accept/reject decision. If all calls {\tt FAIL}, we increment the guess for the running time $t(n)$ and try again.

Our {\sc Tree Evaluation} instance $R_{G'}$ has height at most $B = \lceil t(n)/c \rceil$, where each inner node has at most $(2^d +1)\cdot p$ children, working on blocks of bitlength $2^d \cdot b' \leq O(c^d + d \log t)$. Applying the Cook-Mertz algorithm for {\sc Tree Evaluation}, recalling that $d$ and $p$ are both constants, and including the $O(B)$ bits we use to encode a candidate graph $G'$, we obtain a simulation running in space 
\[O\left(c^d + d \log t + \frac{t(n)}{c} \cdot \log(c^d)\right) \leq O\left(c^d + \frac{t(n)}{c} \cdot \log t(n)\right),\] since $c \leq t(n)$ and $d$ is constant. 

Setting $c = (t(n) \log t(n))^{1/(d+1)}$, the resulting space bound is $O((t(n) \log t(n))^{1-1/(d+1)})$. 
\end{proof}

We remark that putting {\sc Tree Evaluation} in ${\sf LOGSPACE}$ would also directly improve the above simulation as well, to use $O(t(n)^{1-1/(d+1)})$ space. 

\section{Some Consequences}\label{sec:consequences}

First, we describe some simple lower bounds that follow by diagonalization. We will use the following form of the space hierarchy theorem:

\begin{theorem}[\cite{DBLP:conf/focs/StearnsHL65}] 
\label{thm:sp-hier} For space constructible $s'(n), s(n) \geq n$ such that $s'(n) < o(s(n))$, we have 
$\SPACE[s(n)] \not\subset \SPACE[s'(n)]$.
\end{theorem}

\begin{reminder}{Corollary~\ref{cor:space-lb}} For space constructible $s(n) \geq n$ and all $\eps > 0$,
$\SPACE[s(n)] \not\subset \TIME[s(n)^{2-\eps}]$.
\end{reminder}

\begin{proof} Assume to the contrary that $\SPACE[s(n)] \subseteq \TIME[s(n)^{2-\eps}]$ for some $\eps > 0$ and some space constructible $s(n)$. By \Cref{thm:main}, we have
\[\SPACE[s(n)] \subseteq \TIME[s(n)^{2-\eps}] \subseteq \SPACE[(s(n)^{2-\eps} \log s(n))^{1/2}] \subseteq \SPACE[s'(n)]\] for a function $s'(n)$ which is $o(s(n))$. This contradicts \Cref{thm:sp-hier}.
\end{proof}

\begin{reminder}{\Cref{cor:quad-lb}}
The language $L = \{\langle M,x,1^k\rangle \mid |M| \leq k\text{ and }M(x)\text{ halts in space } k\}$ requires $n^{2-\eps}$ time to solve on a multitape Turing machine, for every $\eps  > 0$.    
\end{reminder}

\begin{proof} Suppose $L$ can be solved in $n^{2-\eps}$ time for some $\eps > 0$. We show every language in $\SPACE[n]$ could then be solved in time $O(n^{2-\eps})$, contradicting \Cref{cor:space-lb}. Let $L' \in \SPACE[n]$ be decided by a Turing machine $M$ using space $c n$ for a constant $c \geq 1$. Given an input $x$ to $M$, we call our $n^{2-\eps}$ time algorithm for $L$ on the input $\langle M,x,1^{c |x|}\rangle$ of length $n = O(|x|)$. This takes $O(|x|^{2-\eps})$ time, a contradiction.
\end{proof}

Observe the same proof also shows a time lower bound of the form $n^2/\log^c n$, for a constant $c > 0$.

\subsection{Subexponential Size Branching Programs for Circuits} 
\label{sec:bp}

Here, we observe that \Cref{thm:main} implies that subquadratic size circuits can be simulated with subexponential size branching programs:

\begin{reminder}{\Cref{cor:bp}}
There is a universal $k \geq 1$ such that for all $s \geq n$, every bounded fan-in circuit of size $s$ and $n$ inputs has a branching program of size at most $2^{k \sqrt{s} \log^k s}$.
\end{reminder}

Recall the {\sc Circuit Evaluation} problem: given the description of a Boolean circuit $C$ of fan-in two with one output, and given an input $x$, does $C$ on $x$ evaluate to $1$? We assume $C$ is encoded in \emph{topological order}: the gates are numbered $1,\ldots,s$, and for all $i \in [s]$, the $i$-th gate in the encoding only takes inputs from gates appearing earlier in the encoding. For simplicity, we further assume each gate $i$ provides a correct list of all future gates $j_1,\ldots,j_k > i$ that will consume the output of gate $i$. Note that when $s(n)$ is at least the number of input bits, we can still afford an encoding of size-$s(n)$ circuits in $O(s(n) \log s(n))$ bits, carrying this extra information. An old result commonly credited to Pippenger is that {\sc Circuit Evaluation} on topologically-ordered circuits can be efficiently solved on multitape Turing machines:
\begin{theorem}[Pippenger~\cite{Pippenger77}]\label{thm:ce}
{\sc Circuit Evaluation} $\in \TIME[n \cdot \poly(\log n)]$.
\end{theorem}
The reference~\cite{Pippenger77} is difficult to find, however the PhD thesis of Swamy~(\cite[Chapter 2, Theorem 2.1]{DBLP:phd/us/Swamy78}) gives an exposition of an $O(n \cdot (\log n)^3)$ time bound.\footnote{Available at \url{http://static.cs.brown.edu/research/pubs/theses/phd/1978/swamy.pdf}.}  Swamy describes his construction as simulating oblivious RAMs / straight-line programs; it readily applies to circuits. The high-level idea is to solve a more general problem: given the description of a circuit $C$ and input $x$, we wish to insert the output values of each gate of $C(x)$ directly into the description of $C$. To solve this problem on $C$ of size $s$, we first recursively evaluate the circuit on the first $\lfloor s/2 \rfloor$ gates (in topological order) which sets values to all those gates. We pass over those values, and collect the outputs of all gates among the first $\lfloor s/2 \rfloor$ that will be used as inputs to the remaining $\lceil s/2 \rceil$ gates. We sort these outputs by gate index, then recursively evaluate the remaining $\lceil s/2 \rceil$ gates on $x$ plus the list of outputs. This leads to a runtime recurrence of $T(n) \leq 2 \cdot T(n/2) + O(n \cdot (\log n)^2)$ (using the fact that $n = \Theta(s \log s)$). 

Combining \Cref{thm:ce} and \Cref{thm:main}, we directly conclude that {\sc Circuit Evaluation} can be solved in $\sqrt{n}\cdot \poly(\log n)$ space. Now, given any circuit $C$ of size $s$, we hardcode its description into the input of a multitape Turing machine using $s'(n) = \sqrt{s}\cdot \poly(\log s)$ space for {\sc Circuit Evaluation}. Applying the standard translation of $s'(n)$-space algorithms into branching programs of $2^{O(s'(n))}$ size (see for instance the survey~\cite{DBLP:conf/fct/Razborov91}), \Cref{cor:bp} follows.

\section{Discussion} \label{sec:discussion}

We have shown that multitape Turing machines and circuits have surprisingly space-efficient evaluation algorithms, via a reduction to the {\sc Tree Evaluation} problem. Two reflective remarks come to mind. 

First, we find it very interesting that {\sc Tree Evaluation}, which was originally proposed and studied in the hopes of getting a handle on ${\sf LOGSPACE} \neq \P$, may turn out to be more useful for making progress on $\P \neq \PSPACE$. In any case, it is clear that {\sc Tree Evaluation} is a central problem in complexity theory.

Second, we find it fortunate that the main reduction of this paper (from time-$t$ multitape Turing machine computations to {\sc Tree Evaluation}) was found \emph{after} the Cook-Mertz procedure was discovered. Had our reduction been found first, the community (including the author) would have likely declared the following theorem (a cheeky alternative way of presenting the main reduction of \Cref{thm:main}) as a ``barrier'' to further progress on {\sc Tree Evaluation}, and possibly discouraged work on the subject:

\begin{quote}
``{\bf Theorem}.'' Unless the 50-year-old \cite{DBLP:journals/jacm/HopcroftPV77} simulation $\TIME[t] \subseteq \SPACE[t/\log t]$ can be improved, {\sc Tree Evaluation} instances of constant arity, height $h$, and $b$-bit values cannot be solved in $o(h \cdot b / \log(h\cdot b))$ space.\footnote{Here, we think of $o(h \cdot b / \log(h\cdot b))$ as shorthand for $O(h \cdot b/(f(h \cdot b) \cdot \log(h \cdot b)))$ for some unbounded function $f$.}
\end{quote}

\Cref{thm:main} and its relatives open up an entirely new set of questions that did not seem possible to ask before. Here are a few tantalizing ones.

\smallskip

{\bf Can the simulation of \Cref{thm:main} be improved to show $\TIME[t] \subseteq \SPACE[\sqrt{t}]$?} We discussed some prospects for a yes-answer in \Cref{sec:log} (\emph{e.g.}, showing {\sc Tree Evaluation} is in ${\sf LOGSPACE}$).  An interesting secondary question is whether the longstanding $O(\sqrt{t})$-space simulation of \emph{one tape} time-$t$ Turing machines~\cite{DBLP:journals/jacm/HopcroftU68a,DBLP:journals/jcss/Paterson72} can be improved to $O(t^{1/2-\eps})$ space, for some $\eps > 0$. 

\smallskip

{\bf Is there an $\eps > 0$ such that $\TIME[t] \subseteq \mathsf{ATIME}[t^{1-\eps}]$?} Is there a better speedup of time $t$, using alternating Turing machines? Recall the best-known simulation is $\TIME[t] \subseteq \mathsf{ATIME}[t/\log t]$~\cite{DBLP:journals/jcss/DymondT85}. A yes-answer to this question would imply (for example) a super-linear time lower bound on solving quantified Boolean formulas (QBF), which is still open~\cite{DBLP:journals/eccc/Williams08,DBLP:journals/cc/LiptonW13}. Note that if \Cref{thm:main} could be improved to $\TIME[t] \subseteq \SPACE[t^{1/2-\eps}]$ for some $\eps > 0$, then we would have $\TIME[t] \subseteq \mathsf{ATIME}[t^{1-2\eps}]$.

\smallskip

{\bf Can time-$t$ random-access Turing machines be simulated in space $O(t^{1-\eps})$, for some $\eps > 0$?} As mentioned in \Cref{rem:RAM}, the answer is yes for ``oblivious'' models in which data access patterns can be calculated in advance. In both \Cref{thm:main} and \Cref{thm:d-dimensional}, we exploit the \emph{locality} of low-dimensional tape storage: without that property, the indegrees of nodes in the computation graph would be rather high, and (as far as we can tell) the resulting simulation would not improve the known $O(t/\log t)$ space bounds for simulating $t$ time in random-access models~\cite{DBLP:journals/jcss/PaulR81,DBLP:journals/mst/HalpernLMW86}. 
Similarly to the previous question, if $\TIME[t] \subseteq \SPACE[t^{1/2-\eps}]$ for some $\eps > 0$, then the answer to the question would be yes, since for natural random-access models (e.g., log-cost RAMs), time $t$ can be simulated by $O(t^2)$-time multitape Turing machines~\cite{DBLP:journals/jacm/PippengerF79}. A similar statement holds for improving the $d$-dimensional simulation of \Cref{thm:d-dimensional}~\cite{DBLP:journals/siamcomp/Loui83}. On the other hand, if the answer to the question is no, then we would separate linear time for multitape Turing machines and linear time for random-access models, another longstanding open question (see for example~\cite{GrandjeanS02}).

\smallskip

{\bf Is a time-space tradeoff possible? For example, is $\TIME[t] \subseteq {\sf TIMESPACE}[2^{\tilde{O}(t^{\eps})}, \tilde{O}(t^{1-\eps})]$, for all $\eps > 0$?} The Cook-Mertz procedure needs to compute a low-degree extension of the function at each node, which is time-consuming: for time blocks of $b$ steps, it takes $2^{\Theta(b)}$ time to compute the low-degree extension at a given point. If low-degree extensions of time-$t$ computations could be computed more time-efficiently, then such a time-space tradeoff may be possible. Note that if the \emph{multilinear} extension of a given CNF on $n$ variables and $m$ clauses could be evaluated over a large field in $1.999^n \cdot 2^{o(m)}$ time, then the Strong Exponential Time Hypothesis~\cite{DBLP:journals/jcss/ImpagliazzoP01,DBLP:conf/iwpec/CalabroIP09} would be false: the multilinear extension is the identically-zero polynomial if and only if the CNF is unsatisfiable, so one could use DeMillo–Lipton–Schwartz–Zippel~\cite{DBLP:journals/ipl/DemilloL78} to test for unsatisfiability. However, this observation does not apparently rule out the possibility of evaluating low-degree but non-multilinear extensions efficiently.

\smallskip

{\bf Can the simulation be applied recursively, to show that $\TIME[t] \subseteq \SPACE[t^{\eps}]$ for all $\eps > 0$?} Note that a yes-answer would imply $\P \neq \PSPACE$. At first glance, a recursive extension of \Cref{thm:main} seems natural: we decompose a time-$t$ computation into $O(t/b)$ blocks, each of which are time-$b$ computations. (This property \emph{was} successfully exploited recursively in~\cite{DBLP:journals/cc/LiptonW13}, for example, to show some weak lower bounds on QBF.) However, we cannot directly apply recursion to time blocks, because in the Cook-Mertz procedure, the task at each function evaluation is not just to simulate for $b$ steps, but to compute a low-degree extension of the function (as noted in the previous question). If low-degree extensions of time-$t$ computations can be computed in small space, then there is a possibility of recursion. However, given the discussion on the previous question, there is reason to think this will be difficult.

\smallskip

{\bf Is there a barrier to further progress?} Perhaps lower bounds for the Node-Named Jumping Automata on Graphs (NNJAG) model (a popular model for restricted space lower bounds) \cite{DBLP:conf/focs/Poon93,DBLP:journals/tcs/Poon00,EdmondsPA99} could demonstrate a barrier. Many complex algorithms such as Reingold's~\cite{DBLP:journals/jacm/Reingold08} can be simulated in the NNJAG model~\cite{DBLP:conf/isaac/LuZPC05}; does the Cook-Mertz algorithm have this property as well? We believe the answer is probably no: NNJAG is fundamentally a node-pebbling model, and the Cook-Mertz procedure definitely breaks space lower bounds for pebbling DAGs and trees~\cite{DBLP:journals/jacm/LengauerT82,DBLP:journals/toct/CookMWBS12}. Pebbling lower bounds were a major bottleneck to improving \cite{DBLP:journals/jacm/HopcroftPV77}.

While \Cref{thm:main} and \Cref{thm:warmup} are non-relativizing (see for example~\cite{DBLP:series/wsscs/HartmanisCCRR93}), the simulations do permit a \emph{restricted} form of relativization, as do all prior simulations of time in smaller space. Define ${\sf TIME\text{-}LENGTH}^A[t(n),\ell(n)]$ to be the class of problems solvable in $O(t(n))$ time with queries to the oracle $A$, where all oracle queries are restricted to have length at most $\ell(n)$. Similarly define ${\sf SPACE\text{-}LENGTH}^A[s(n),\ell(n)]$. We observe the following extension of \Cref{thm:warmup}:
\begin{theorem}
    For every oracle $A$ and $t(n) \geq n$,
    ${\sf TIME\text{-}LENGTH}^A[t(n),\sqrt{t(n)} \log t(n)]$ is contained in ${\sf SPACE\text{-}LENGTH}^A[\sqrt{t(n)} \log t(n),\sqrt{t(n)} \log t(n)]$.
\end{theorem}
The theorem follows because each oracle query can be arranged to be written in a single tape block of length $O(\sqrt{t(n)} \log t(n))$, and the Cook-Mertz procedure treats the evaluation functions as \emph{black boxes}. Thus each evaluation function $F_{h,i}$ can simply call the oracle $A$ as needed. (Tretkoff~\cite{DBLP:conf/coco/Tretkoff86} made a similar observation for the 
simulation of $\TIME[t]$ in $\Sigma_2 \TIME[o(t)]$, of Paul-Pippenger-Szemer\'{e}di-Trotter~\cite{DBLP:conf/focs/PaulPST83}. Such a length-restricted oracle model was also studied in~\cite{DBLP:journals/algorithmica/CaiW06}.) Is there a barrier to improving such (restricted) relativizing results to obtain $\P \neq \PSPACE$? This seems related to the fact that Cai and Watanabe~\cite{DBLP:journals/algorithmica/CaiW06} were unable to collapse $\PSPACE$ to $\P$ with ``random access to advice'' (length-restricted access to oracles).

\paragraph{Acknowledgments.} I am grateful to Rahul Ilango, Egor Lifar, Priya Malhotra, Danil Sibgatullin, and Virginia Vassilevska Williams for valuable discussions on {\sc Tree Evaluation} and the Cook-Mertz procedure. I am also grateful to Shyan Akmal, Paul Beame, Lijie Chen, Ce Jin, Dylan McKay, and the anonymous STOC reviewers for helpful comments on an earlier draft of this paper. 

\bibliography{main}

\newcommand{\etalchar}[1]{$^{#1}$}
\begin{thebibliography}{FLvMV05}

\bibitem[AB09]{AB09}
Sanjeev Arora and Boaz Barak.
\newblock {\em Computational Complexity - {A} Modern Approach}.
\newblock Cambridge University Press, 2009.
\newblock URL: \url{http://www.cambridge.org/catalogue/catalogue.asp?isbn=9780521424264}.

\bibitem[AL81]{DBLP:journals/mst/AdlemanL81}
Leonard~M. Adleman and Michael~C. Loui.
\newblock Space-bounded simulation of multitape turing machines.
\newblock {\em Math. Syst. Theory}, 14:215--222, 1981.
\newblock URL: \url{https://doi.org/10.1007/BF01752397}.

\bibitem[CIP09]{DBLP:conf/iwpec/CalabroIP09}
Chris Calabro, Russell Impagliazzo, and Ramamohan Paturi.
\newblock The complexity of satisfiability of small depth circuits.
\newblock In {\em Parameterized and Exact Computation, 4th International Workshop, {IWPEC}}, volume 5917 of {\em Lecture Notes in Computer Science}, pages 75--85. Springer, 2009.
\newblock URL: \url{https://doi.org/10.1007/978-3-642-11269-0\_6}.

\bibitem[CM20]{DBLP:conf/stoc/CookM20}
James Cook and Ian Mertz.
\newblock Catalytic approaches to the tree evaluation problem.
\newblock In {\em Proceedings of {STOC}}, pages 752--760. {ACM}, 2020.
\newblock URL: \url{https://doi.org/10.1145/3357713.3384316}.

\bibitem[CM21]{DBLP:journals/eccc/CookM21}
James Cook and Ian Mertz.
\newblock Encodings and the tree evaluation problem.
\newblock {\em Electron. Colloquium Comput. Complex.}, {TR21-054}, 2021.
\newblock URL: \url{https://eccc.weizmann.ac.il/report/2021/054}.

\bibitem[CM22]{DBLP:conf/coco/CookM22}
James Cook and Ian Mertz.
\newblock Trading time and space in catalytic branching programs.
\newblock In {\em 37th Computational Complexity Conference, {CCC}}, volume 234 of {\em LIPIcs}, pages 8:1--8:21. Schloss Dagstuhl - Leibniz-Zentrum f{\"{u}}r Informatik, 2022.
\newblock URL: \url{https://doi.org/10.4230/LIPIcs.CCC.2022.8}.

\bibitem[CM24]{DBLP:conf/stoc/CookM24}
James Cook and Ian Mertz.
\newblock Tree evaluation is in space {O}(log n {\(\cdot\)} log log n).
\newblock In {\em Proceedings of the 56th Annual {ACM} Symposium on Theory of Computing ({STOC})}, pages 1268--1278. {ACM}, 2024.
\newblock URL: \url{https://doi.org/10.1145/3618260.3649664}.

\bibitem[CMW{\etalchar{+}}12]{DBLP:journals/toct/CookMWBS12}
Stephen~A. Cook, Pierre McKenzie, Dustin Wehr, Mark Braverman, and Rahul Santhanam.
\newblock Pebbles and branching programs for tree evaluation.
\newblock {\em {ACM} Trans. Comput. Theory}, 3(2):4:1--4:43, 2012.
\newblock URL: \url{https://doi.org/10.1145/2077336.2077337}.

\bibitem[CW06]{DBLP:journals/algorithmica/CaiW06}
Jin{-}yi Cai and Osamu Watanabe.
\newblock Random access to advice strings and collapsing results.
\newblock {\em Algorithmica}, 46(1):43--57, 2006.
\newblock URL: \url{https://doi.org/10.1007/s00453-006-0078-8}.

\bibitem[DL78]{DBLP:journals/ipl/DemilloL78}
Richard~A. DeMillo and Richard~J. Lipton.
\newblock A probabilistic remark on algebraic program testing.
\newblock {\em Inf. Process. Lett.}, 7(4):193--195, 1978.
\newblock URL: \url{https://doi.org/10.1016/0020-0190(78)90067-4}.

\bibitem[DT85]{DBLP:journals/jcss/DymondT85}
Patrick~W. Dymond and Martin Tompa.
\newblock Speedups of deterministic machines by synchronous parallel machines.
\newblock {\em J. Comput. Syst. Sci.}, 30(2):149--161, 1985.
\newblock URL: \url{https://doi.org/10.1016/0022-0000(85)90011-X}.

\bibitem[EPA99]{EdmondsPA99}
Jeff Edmonds, Chung~Keung Poon, and Dimitris Achlioptas.
\newblock Tight lower bounds for st-connectivity on the {NNJAG} model.
\newblock {\em {SIAM} J. Comput.}, 28(6):2257--2284, 1999.

\bibitem[FLvMV05]{DBLP:journals/jacm/FortnowLMV05}
Lance Fortnow, Richard~J. Lipton, Dieter van Melkebeek, and Anastasios Viglas.
\newblock Time-space lower bounds for satisfiability.
\newblock {\em J. {ACM}}, 52(6):835--865, 2005.
\newblock URL: \url{https://doi.org/10.1145/1101821.1101822}.

\bibitem[Gol08]{Goldreich08}
Oded Goldreich.
\newblock {\em Computational complexity - a conceptual perspective}.
\newblock Cambridge University Press, 2008.
\newblock URL: \url{https://doi.org/10.1017/CBO9780511804106}.

\bibitem[Gol24]{Goldreich-CM24}
Oded Goldreich.
\newblock On the {Cook-Mertz} {Tree Evaluation} procedure.
\newblock {\em Electron. Colloquium Comput. Complex.}, {TR24-109}, 2024.
\newblock URL: \url{https://eccc.weizmann.ac.il/report/2024/109}.

\bibitem[GS02]{GrandjeanS02}
Etienne Grandjean and Thomas Schwentick.
\newblock Machine-independent characterizations and complete problems for deterministic linear time.
\newblock {\em {SIAM} J. Comput.}, 32(1):196--230, 2002.
\newblock \href {https://doi.org/10.1137/S0097539799360240} {\path{doi:10.1137/S0097539799360240}}.

\bibitem[HCC{\etalchar{+}}93]{DBLP:series/wsscs/HartmanisCCRR93}
Juris Hartmanis, Richard Chang, Suresh Chari, Desh Ranjan, and Pankaj Rohatgi.
\newblock Relativization: a revisionistic retrospective.
\newblock In {\em Current Trends in Theoretical Computer Science - Essays and Tutorials}, volume~40 of {\em World Scientific Series in Computer Science}, pages 537--548. World Scientific, 1993.
\newblock URL: \url{https://doi.org/10.1142/9789812794499\_0040}.

\bibitem[HLMW86]{DBLP:journals/mst/HalpernLMW86}
Joseph~Y. Halpern, Michael~C. Loui, Albert~R. Meyer, and Daniel Weise.
\newblock On time versus space {III}.
\newblock {\em Math. Syst. Theory}, 19(1):13--28, 1986.
\newblock URL: \url{https://doi.org/10.1007/BF01704903}.

\bibitem[HPV75]{DBLP:journals/jacm/HopcroftPV77}
John~E. Hopcroft, Wolfgang~J. Paul, and Leslie~G. Valiant.
\newblock On time versus space.
\newblock {\em J. {ACM}}, 24(2):332--337, 1977. Conference version in FOCS'75.
\newblock URL: \url{https://doi.org/10.1145/322003.322015}.

\bibitem[HS66]{DBLP:journals/jacm/HennieS66}
F.~C. Hennie and Richard~Edwin Stearns.
\newblock Two-tape simulation of multitape turing machines.
\newblock {\em J. {ACM}}, 13(4):533--546, 1966.
\newblock URL: \url{https://doi.org/10.1145/321356.321362}.

\bibitem[HU68]{DBLP:journals/jacm/HopcroftU68a}
John~E. Hopcroft and Jeffrey~D. Ullman.
\newblock Relations between time and tape complexities.
\newblock {\em J. {ACM}}, 15(3):414--427, 1968.
\newblock URL: \url{https://doi.org/10.1145/321466.321474}.

\bibitem[IM83]{DBLP:journals/siamcomp/IbarraM83}
Oscar~H. Ibarra and Shlomo Moran.
\newblock Some time-space tradeoff results concerning single-tape and offline {TM's}.
\newblock {\em {SIAM} J. Comput.}, 12(2):388--394, 1983.
\newblock URL: \url{https://doi.org/10.1137/0212025}.

\bibitem[IP01]{DBLP:journals/jcss/ImpagliazzoP01}
Russell Impagliazzo and Ramamohan Paturi.
\newblock On the complexity of k-sat.
\newblock {\em J. Comput. Syst. Sci.}, 62(2):367--375, 2001.
\newblock URL: \url{https://doi.org/10.1006/jcss.2000.1727}, \href {https://doi.org/10.1006/JCSS.2000.1727} {\path{doi:10.1006/JCSS.2000.1727}}.

\bibitem[Juk12]{Jukna}
Stasys Jukna.
\newblock {\em Boolean Function Complexity - Advances and Frontiers}, volume~27 of {\em Algorithms and combinatorics}.
\newblock Springer, 2012.
\newblock URL: \url{https://doi.org/10.1007/978-3-642-24508-4}.

\bibitem[KLRS12]{DBLP:journals/tcs/KalyanasundaramLRS12}
Subrahmanyam Kalyanasundaram, Richard~J. Lipton, Kenneth~W. Regan, and Farbod Shokrieh.
\newblock Improved simulation of nondeterministic turing machines.
\newblock {\em Theor. Comput. Sci.}, 417:66--73, 2012.
\newblock URL: \url{https://doi.org/10.1016/j.tcs.2011.05.018}.

\bibitem[KLV03]{DBLP:journals/tcs/KarakostasLV03}
George Karakostas, Richard~J. Lipton, and Anastasios Viglas.
\newblock On the complexity of intersecting finite state automata and {NL} versus {NP}.
\newblock {\em Theor. Comput. Sci.}, 302(1-3):257--274, 2003.
\newblock URL: \url{https://doi.org/10.1016/S0304-3975(02)00830-7}.

\bibitem[Kur64]{Kuroda64}
S.-Y. Kuroda.
\newblock Classes of languages and linear-bounded automata.
\newblock {\em Information and control}, 7(2):207--223, 1964.

\bibitem[LL90]{DBLP:journals/siamcomp/LiskiewiczL90}
Maciej Liskiewicz and Krzysztof Lorys.
\newblock Fast simulations of time-bounded one-tape turing machines by space-bounded ones.
\newblock {\em {SIAM} J. Comput.}, 19(3):511--521, 1990.
\newblock URL: \url{https://doi.org/10.1137/0219034}.

\bibitem[Lou81]{DBLP:journals/tcs/Loui81}
Michael~C. Loui.
\newblock A space bound for one-tape multidimensional turing machines.
\newblock {\em Theor. Comput. Sci.}, 15:311--320, 1981.
\newblock URL: \url{https://doi.org/10.1016/0304-3975(81)90084-0}.

\bibitem[Lou83]{DBLP:journals/siamcomp/Loui83}
Michael~C. Loui.
\newblock Optimal dynamic embedding of trees into arrays.
\newblock {\em {SIAM} J. Comput.}, 12(3):463--472, 1983.
\newblock \href {https://doi.org/10.1137/0212030} {\path{doi:10.1137/0212030}}.

\bibitem[LT82]{DBLP:journals/jacm/LengauerT82}
Thomas Lengauer and Robert~Endre Tarjan.
\newblock Asymptotically tight bounds on time-space trade-offs in a pebble game.
\newblock {\em J. {ACM}}, 29(4):1087--1130, 1982.
\newblock \href {https://doi.org/10.1145/322344.322354} {\path{doi:10.1145/322344.322354}}.

\bibitem[LW13]{DBLP:journals/cc/LiptonW13}
Richard~J. Lipton and Ryan Williams.
\newblock Amplifying circuit lower bounds against polynomial time, with applications.
\newblock {\em Comput. Complex.}, 22(2):311--343, 2013.
\newblock URL: \url{https://doi.org/10.1007/s00037-013-0069-5}.

\bibitem[LZPC05]{DBLP:conf/isaac/LuZPC05}
Pinyan Lu, Jialin Zhang, Chung~Keung Poon, and Jin{-}yi Cai.
\newblock Simulating undirected \emph{st}-connectivity algorithms on uniform {JAGs} and {NNJAGs}.
\newblock In {\em Proceedings of 16th International Symposium on Algorithms and Computation {(ISAAC)}}, volume 3827 of {\em Lecture Notes in Computer Science}, pages 767--776. Springer, 2005.
\newblock URL: \url{https://doi.org/10.1007/11602613\_77}.

\bibitem[MW17]{DBLP:conf/coco/MurrayW17}
Cody~D. Murray and R.~Ryan Williams.
\newblock Easiness amplification and uniform circuit lower bounds.
\newblock In {\em 32nd Computational Complexity Conference ({CCC})}, volume~79 of {\em LIPIcs}, pages 8:1--8:21. Schloss Dagstuhl - Leibniz-Zentrum f{\"{u}}r Informatik, 2017.
\newblock URL: \url{https://doi.org/10.4230/LIPIcs.CCC.2017.8}.

\bibitem[NW94]{DBLP:journals/jcss/NisanW94}
Noam Nisan and Avi Wigderson.
\newblock Hardness vs randomness.
\newblock {\em J. Comput. Syst. Sci.}, 49(2):149--167, 1994.
\newblock URL: \url{https://doi.org/10.1016/S0022-0000(05)80043-1}.

\bibitem[Pat72]{DBLP:journals/jcss/Paterson72}
Mike Paterson.
\newblock Tape bounds for time-bounded turing machines.
\newblock {\em J. Comput. Syst. Sci.}, 6(2):116--124, 1972.
\newblock URL: \url{https://doi.org/10.1016/S0022-0000(72)80017-5}.

\bibitem[PF79]{DBLP:journals/jacm/PippengerF79}
Nicholas Pippenger and Michael~J. Fischer.
\newblock Relations among complexity measures.
\newblock {\em J. {ACM}}, 26(2):361--381, 1979.
\newblock URL: \url{https://doi.org/10.1145/322123.322138}.

\bibitem[Pip77]{Pippenger77}
Nicholas Pippenger.
\newblock Fast simulation of combinational logic networks by machines without random-access storage.
\newblock In {\em Proceedings of the Fifteenth Annual Allerton Conference on Communication, Control and Computing}, pages 25--33, 1977.

\bibitem[Poo93]{DBLP:conf/focs/Poon93}
Chung~Keung Poon.
\newblock Space bounds for graph connectivity problems on node-named jags and node-ordered jags.
\newblock In {\em 34th Annual Symposium on Foundations of Computer Science {(FOCS)}}, pages 218--227. {IEEE} Computer Society, 1993.
\newblock URL: \url{https://doi.org/10.1109/SFCS.1993.366865}.

\bibitem[Poo00]{DBLP:journals/tcs/Poon00}
Chung~Keung Poon.
\newblock A space lower bound for st-connectivity on node-named jags.
\newblock {\em Theor. Comput. Sci.}, 237(1-2):327--345, 2000.
\newblock URL: \url{https://doi.org/10.1016/S0304-3975(00)00019-0}.

\bibitem[PPST83]{DBLP:conf/focs/PaulPST83}
Wolfgang~J. Paul, Nicholas Pippenger, Endre Szemer{\'{e}}di, and William~T. Trotter.
\newblock On determinism versus non-determinism and related problems (preliminary version).
\newblock In {\em 24th Annual Symposium on Foundations of Computer Science ({FOCS})}, pages 429--438. {IEEE} Computer Society, 1983.

\bibitem[PR80]{DBLP:journals/acta/PaulR80}
Wolfgang~J. Paul and R{\"{u}}diger Reischuk.
\newblock On alternation {II.} {A} graph theoretic approach to determinism versus nondeterminism.
\newblock {\em Acta Informatica}, 14:391--403, 1980.
\newblock URL: \url{https://doi.org/10.1007/BF00286494}.

\bibitem[PR81]{DBLP:journals/jcss/PaulR81}
Wolfgang~J. Paul and R{\"{u}}diger Reischuk.
\newblock On time versus space {II.}
\newblock {\em J. Comput. Syst. Sci.}, 22(3):312--327, 1981.
\newblock URL: \url{https://doi.org/10.1016/0022-0000(81)90035-0}.

\bibitem[PV76]{DBLP:journals/tcs/PatersonV76}
Mike Paterson and Leslie~G. Valiant.
\newblock Circuit size is nonlinear in depth.
\newblock {\em Theor. Comput. Sci.}, 2(3):397--400, 1976.
\newblock URL: \url{https://doi.org/10.1016/0304-3975(76)90090-6}.

\bibitem[Raz91]{DBLP:conf/fct/Razborov91}
Alexander~A. Razborov.
\newblock Lower bounds for deterministic and nondeterministic branching programs.
\newblock In Lothar Budach, editor, {\em 8th International Symposium on Fundamentals of Computation Theory ({FCT})}, volume 529 of {\em Lecture Notes in Computer Science}, pages 47--60. Springer, 1991.
\newblock URL: \url{https://doi.org/10.1007/3-540-54458-5\_49}.

\bibitem[Rei08]{DBLP:journals/jacm/Reingold08}
Omer Reingold.
\newblock Undirected connectivity in log-space.
\newblock {\em J. {ACM}}, 55(4):17:1--17:24, 2008.
\newblock \href {https://doi.org/10.1145/1391289.1391291} {\path{doi:10.1145/1391289.1391291}}.

\bibitem[SHL65]{DBLP:conf/focs/StearnsHL65}
Richard~Edwin Stearns, Juris Hartmanis, and Philip~M. {Lewis II}.
\newblock Hierarchies of memory limited computations.
\newblock In {\em 6th Annual Symposium on Switching Circuit Theory and Logical Design}, pages 179--190. {IEEE} Computer Society, 1965.
\newblock URL: \url{https://doi.org/10.1109/FOCS.1965.11}.

\bibitem[Sip88]{DBLP:journals/jcss/Sipser88}
Michael Sipser.
\newblock Expanders, randomness, or time versus space.
\newblock {\em J. Comput. Syst. Sci.}, 36(3):379--383, 1988.
\newblock URL: \url{https://doi.org/10.1016/0022-0000(88)90035-9}.

\bibitem[SM73]{DBLP:conf/stoc/StockmeyerM73}
Larry~J. Stockmeyer and Albert~R. Meyer.
\newblock Word problems requiring exponential time: Preliminary report.
\newblock In {\em Proceedings of {STOC}}, pages 1--9. {ACM}, 1973.
\newblock URL: \url{https://doi.org/10.1145/800125.804029}.

\bibitem[SSZ98]{DBLP:journals/jacm/SaksSZ98}
Michael~E. Saks, Aravind Srinivasan, and Shiyu Zhou.
\newblock Explicit or-dispersers with polylogarithmic degree.
\newblock {\em J. {ACM}}, 45(1):123--154, 1998.
\newblock URL: \url{https://doi.org/10.1145/273865.273915}.

\bibitem[SvEB88]{DBLP:journals/iandc/SlotB88}
Cees~F. Slot and Peter van Emde~Boas.
\newblock The problem of space invariance for sequential machines.
\newblock {\em Inf. Comput.}, 77(2):93--122, 1988.
\newblock URL: \url{https://doi.org/10.1016/0890-5401(88)90052-1}.

\bibitem[Swa78]{DBLP:phd/us/Swamy78}
Sowmitri Swamy.
\newblock {\em On Space-Time Tradeoffs}.
\newblock PhD thesis, Brown University, {USA}, 1978.
\newblock URL: \url{https://cs.brown.edu/research/pubs/theses/phd/1978/swamy.pdf}.

\bibitem[Tre86]{DBLP:conf/coco/Tretkoff86}
Carol Tretkoff.
\newblock Bounded oracles and complexity classes inside linear space.
\newblock In Alan~L. Selman, editor, {\em Proceedings of Structure in Complexity Theory}, volume 223 of {\em Lecture Notes in Computer Science}, pages 347--361. Springer, 1986.
\newblock URL: \url{https://doi.org/10.1007/3-540-16486-3\_110}.

\bibitem[vEB90]{DBLP:books/el/leeuwen90/Boas90}
Peter van Emde~Boas.
\newblock Machine models and simulation.
\newblock In Jan van Leeuwen, editor, {\em Handbook of Theoretical Computer Science, Volume {A:} Algorithms and Complexity}, pages 1--66. Elsevier and {MIT} Press, 1990.

\bibitem[Wil08]{DBLP:journals/eccc/Williams08}
Ryan Williams.
\newblock Non-linear time lower bound for (succinct) quantified boolean formulas.
\newblock {\em Electron. Colloquium Comput. Complex.}, {TR08-076}, 2008.
\newblock URL: \url{https://eccc.weizmann.ac.il/eccc-reports/2008/TR08-076/index.html}.

\end{thebibliography}
\bibliographystyle{alphaurl}

\appendix

\section{Appendix: An Overview of The Cook-Mertz Procedure}

\label{appendix:cm}

The goal of this appendix is to describe the Cook-Mertz procedure for {\sc Tree Evaluation}, and how it extends to arbitrary $d$-ary trees of maximum height $h$, and not just \emph{complete} $d$-ary trees of height $h$, with the full set of $(d^h - 1)/(d-1)$ possible nodes. In particular, any tree with every inner node having \emph{at most} $d$ children and depth \emph{at most} $h$ can be evaluated using their algorithm. 

\begin{reminder}{\Cref{thm:cm}} {\rm \cite[Theorem 7]{DBLP:conf/stoc/CookM24}}  {\sc Tree Evaluation} on trees of bit-length $b$, \emph{maximum} height $h$, and fan-in \emph{at most} $d$, can be computed in $O(d \cdot b + h \log (d \cdot b))$ space.
\end{reminder}

Our description below \emph{heavily} draws from Goldreich's exposition of the Cook-Mertz procedure~\cite{Goldreich-CM24}. Let us stress that we make absolutely no claims of originality in the exposition below; we are only including the following in order to make our paper more self-contained.

First, let us describe how to handle inner nodes in the tree which have less than $d$ children. Letting $\deg(u)$ denote the number of children of $u$, for all $j \in [b]$ we let $f_{u,j} : \{0,1\}^{\deg(u) \cdot b}\rightarrow \{0,1\}$ be the function which returns the $j$-th bit of the function $f_u$ computed at node $u$. For every node $u$ with $\deg(u) < d$, we add extra children to $u$ which are simply leaves with value $0^b$, so that $u$ has exactly $d$ children. The resulting new functions $f_{u,j} : \{0,1\}^{d \cdot b} \rightarrow \{0,1\}$ at node $u$ are defined to simply ignore all bits of input with index larger than $\deg(u) \cdot b$. Now all inner nodes of the tree have exactly $d$ children, and all values of inner nodes remain the same as before. For our intended application to an \emph{implicit} {\sc Tree Evaluation} instance, we note that the reduction of this paragraph can also be implemented in a space-efficient way: we only need to be able to check the number of children of the current node $u$, which we can easily do in our applications (e.g., \Cref{thm:warmup} and \Cref{thm:d-dimensional}). In particular, in our reduction the memory stores the entire computation graph which defines the tree, so we know the number of children of every tree node.

Let $\F$ be a field of characteristic two such that $|\F| \geq db^2$. For a node $u$ and index $j \in [b]$, let $\widetilde{f}_{u,j}$ be the multilinear extension (over $\F$) of the function $f_{u,j} : \{0,1\}^{d \cdot b}\rightarrow \{0,1\}$ which returns the $j$-th bit of the function $f_u$ computed at node $u$. In particular, $\widetilde{f}_{u,j}$ is a polynomial of degree $d\cdot b$ on $d\cdot b$ variables. Letting $\vec{x_i}$ denote a block of $b$ variables for $i=1,\ldots,d$,
\begin{align}\label{eq:multilinear}
\widetilde{f}_{u,j}(\vec{x_1},\ldots,\vec{x_d}) = \sum_{a_1,\ldots,a_d \in \{0,1\}^b} \chi_{a_1,\ldots,a_d}(\vec{x_1},\ldots,\vec{x_d}) \cdot f_{u,j}(a_1,\ldots,a_d),
\end{align}
where $\chi_{a_1,\ldots,a_d}(\vec{x_1},\ldots,\vec{x_d})$ is the unique multilinear polynomial of degree $d \cdot b$ such that 
\[\chi_{a_1,\ldots,a_d}(a_1,\ldots,a_d) = 1,\] and $\chi_{a_1,\ldots,a_d}(a'_1,\ldots,a'_d) = 0$ for all $(a'_1,\ldots,a'_d) \in (\{0,1\}^b)^{d}$ such that $(a'_1,\ldots,a'_d) \neq (a_1,\ldots,a_d)$. 

Let $[d]^{\star}$ be the set of all strings over the alphabet $\{1,\ldots,d\}$ (including the empty string). 
Observe that for every node with label $u \in [d]^{\star}$, its children have labels $u1,\ldots,ud$. Thus by definition, the value of $u$ is 
\begin{align}\label{eqn1}
v_u &= f_u(v_{u1},\ldots,v_{ud}) = (f_{u,1}(v_{u1},\ldots,v_{ud}),\ldots,f_{u,b}(v_{u1},\ldots,v_{ud}))\\
\label{eqn2}
&= (\widetilde{f}_{u,1}(v_{u1},\ldots,v_{ud}),\ldots,\widetilde{f}_{u,b}(v_{u1},\ldots,v_{ud})).
\end{align}

The {\sc Tree Evaluation} procedure has two types of storage:
\begin{itemize}
    \item One storage type is ``local'', holding the index of the current node ($O(h \log d)$ bits), as well as a recursion stack of height at most $h$ with $O(\log(d \cdot b))$ bits stored at each level of recursion, which is $O(h \log (d \cdot b)))$ bits in total. 
    
    \item The other type of storage is ``catalytic'' or ``global'' storage (\cite{Goldreich-CM24}). This $O(d \cdot b)$ space is used to store $d+1$ separate $b$-bit blocks. In the final version of the algorithm that we describe, each $b$-bit block corresponds to storing $O(b/\log |\F|)$ elements of $\F$. These blocks are repeatedly reused at every node of the tree, to evaluate the functions at each node. 
    
    In the following, for simplicity, we will describe an $O(d \cdot b \log(d \cdot b))$ space bound using multilinear extensions; with minor modifications, Cook-Mertz~\cite{DBLP:conf/stoc/CookM24} (see also Goldreich~\cite{Goldreich-CM24}) show how relaxing to merely \emph{low-degree} extensions allows us to use only $O(d \cdot b)$ space. At the end, we will discuss how to achieve this space bound.
\end{itemize} 
At any point in time, we will denote the entire storage content of the algorithm by
$(u,\hat{x}_1,\ldots,\hat{x}_d,\hat{y})$, where
\begin{itemize}
\item $u \in [d]^{\star}$, $|u| \leq h$, is the label of a node in the tree, denoting the path from root to that node. ($\eps$ denotes the root node, $i \in [d]$ denotes the $i$-th child of the root, etc.), 
\item each $\hat{x}_i$ is a collection of $b$ registers $\hat{x}^{(1)}_i,\ldots,\hat{x}^{(b)}_i$ holding $b$ elements of $\F$, and
\item $\hat{y}$ is a collection of $b$ registers $\hat{y}^{(1)},\ldots,\hat{y}^{(b)}$, also holding $b$ elements of $\F$.
\end{itemize}
Initially, we set the storage content to be 
\[(\eps,\vec{0},\ldots,\vec{0}),\] i.e., all-zeroes, starting at the root node.

For a node label $u$, let $v_u \in \{0,1\}^b$ be the value of $u$; we think of $v_u$ as an element of $\F^b$. Our goal is to compute $v_{\eps}$, the value of the root of the tree. We will give a recursive procedure {\sc Add} which takes storage content $(u,\hat{x}_1,\ldots,\hat{x}_d,\hat{y})$ and returns the content $(u,\hat{x}_1,\ldots,\hat{x}_d,\hat{y}+ v_u)$. That is, {\sc Add} adds the value $v_u$ to the last register, over the field $\F$. Observe that, if {\sc Add} works correctly, then calling {\sc Add} on $(\eps,\vec{0},\ldots,\vec{0})$ will put the value of the root node into the last register.

Now we describe {\sc Add}. Let $m$ be such that $|\F| = m+1 \geq d \cdot b$, and let $\omega$ be an $m$-th root of unity in $\F$. 


\begin{framed}
\noindent {\sc Add}: Given the storage content $(u,\hat{x}_1,\ldots,\hat{x}_d,\hat{y})$,\\
If $u$ is a leaf, look up the value $v_u$ in the input, and return $(u,\hat{x}_1,\ldots,\hat{x}_d,\hat{y} + v_u)$.\\
For $i = 1,\ldots,m$,\\
\hspace*{0.5em} For $r = 1,\ldots,d$,\\
\hspace*{1.2em} Rotate registers, so $(\hat{x}_r,\ldots,\hat{x}_d,\hat{y},\hat{x}'_1,\ldots,\hat{x}'_{r-1})$ shifts to $(\hat{x}_{r+1},\ldots,\hat{x}_d,\hat{y},\hat{x}'_1,\ldots,\hat{x}'_{r-1},\hat{x}_r)$.\\
\hspace*{1.2em} Multiply $\hat{x}_r$ by $\omega^i$, and call {\sc Add} on $(ur,\hat{x}_{r+1},\ldots,\hat{x}_d,\hat{y},\hat{x}'_1,\ldots,\hat{x}'_{r-1},\omega^i \cdot \hat{x}_r)$, \\ 
\hspace*{1,7em} which returns $(ur, \hat{x}_{r+1},\ldots,\hat{x}_d,\hat{y},\hat{x}'_1,\ldots,\hat{x}'_{r-1},\hat{x}'_r)$, where $\hat{x}'_r = \omega^i \cdot \hat{x}_r+v_{ur}$.\\
\hspace*{1.7em} \emph{(here, $ur$ is just the concatenation of the string $u$ with the symbol $r$)}\\
\hspace*{0.5em} \emph{(at this point, the storage has the form: $(ud,\hat{y},\omega^i \cdot \hat{x}_1+v_{u1},\ldots,\omega^i \cdot \hat{x}_d+v_{ud})$)}\\
\hspace*{0.5em} Update $ud$ back to $u$.\\
\hspace*{0.5em} For all $j = 1,\ldots,b$,\\
\hspace*{1.2em} Compute $\widetilde{f}_{u,j}(\omega^i \cdot \hat{x}_1+v_{u1},\ldots,\omega^i \cdot \hat{x}_d+v_{ud})$ using $O(b)$  extra space.\\
\hspace*{1.2em} Update $\hat{y}^{(j)} = \hat{y}^{(j)} +  \widetilde{f}_{u,j}(\omega^i \cdot \hat{x}_1+v_{u1},\ldots,\omega^i \cdot \hat{x}_d+v_{ud})$.\\
\hspace*{1.2em} Erase the $O(b)$ space used to compute $\widetilde{f}_{u,j}$.\\
\hspace*{0.5em} For $r = d,\ldots,1$,\\
\hspace*{1,2em} Call {\sc Add} on $(ur,\hat{x}_{r+1},\ldots,\hat{x}_d,\hat{y},\omega^i \cdot \hat{x}_1 + v_{u1},\ldots,\omega^i \cdot \hat{x}_r + v_{ur})$, \\ 
\hspace*{1,7em} which returns $(ur,\hat{x}_{r+1},\ldots,\hat{x}_d,\hat{y},\omega^i \cdot \hat{x}_1 + v_{u1},\ldots,\omega^i \cdot \hat{x}_{r-1} + v_{u(r-1)},\hat{x}''_r)$, where $\hat{x}''_r = \omega^i \cdot \hat{x}_r$.\\
\hspace*{1.7em} \emph{(note: here we use the fact that $\F$ is characteristic two)}\\
\hspace*{1.2em} Divide $\hat{x}''_r$ by $\omega^i$, so that $\hat{x}''_r = \hat{x}_r$.\\
\hspace*{1.2em} Rotate registers, so $(\hat{x}_{r+1},\ldots,\hat{x}_d,\hat{y},\omega^i \cdot \hat{x}_1 + v_{u1},\ldots,\omega^i \cdot \hat{x}_{r-1} + v_{u(r-1)},\hat{x}_r)$ shifts to \\
\hspace*{2.2em} $(\hat{x}_r,\hat{x}_{r+1},\ldots,\hat{x}_d,\hat{y},\omega^i \cdot \hat{x}_1 + v_{u1},\ldots,\omega^i \cdot \hat{x}_{r-1} + v_{u(r-1)})$.\\
\hspace*{0.5em} Update $u1$ back to $u$.\\
\emph{(claim: now the storage has the form $(u,\hat{x}_1,\ldots,\hat{x}_d,\hat{y}+ \widetilde{f}_{u}(v_{u1},\ldots,v_{ud}))$)}\\
\emph{(note that $v_{u} = \widetilde{f}_{u}(v_{u1},\ldots,v_{ud})$, by \eqref{eqn1} and \eqref{eqn2})}\\
Return $(u,\hat{x}_1,\ldots,\hat{x}_d,\hat{y} + v_u)$.
\end{framed}

Recall that $\F$ is characteristic two, so that adding $v_{ur}$ twice to the register $\hat{x}_r$ has a net contribution of zero. The key to the correctness of {\sc Add} (and the \emph{claim} in the pseudocode) is the following polynomial interpolation formula (proved in \cite{DBLP:conf/stoc/CookM24}): for every $m$-th root of unity $\omega$ over $\F$, for every $\hat{x}_1,\ldots,\hat{x}_d, v_{u1},\ldots,v_{ud} \in \F^b$, and for every polynomial $P$ of degree less than $m$,
\begin{align}\label{eq:interp}
\sum_{i=1}^m P(\omega^i \cdot \hat{x}_1 + v_{u1},\ldots,\omega^i \cdot \hat{x}_d + v_{ud}) = P(v_{u1},\ldots,v_{ud}).
\end{align}
(Note that our formula is slightly simpler than Cook-Mertz~\cite{DBLP:conf/stoc/CookM24} and Goldreich~\cite{Goldreich-CM24}, because we assume $\F$ is a field of characteristic two.)
Equation~\eqref{eq:interp} ensures that the content of $\hat{y}$ is indeed updated to be $\hat{y} + \widetilde{f_u}(v_{u1},\ldots,v_{ud})$, which equals $\hat{y} + v_u$ by equations \eqref{eqn1} and \eqref{eqn2}.

Let us briefly compare {\sc Add} to the ``obvious'' algorithm for {\sc Tree Evaluation}. The ``obvious'' algorithm allocates fresh new space for each recursive call to store the values of the children, traversing the tree in a depth-first manner. Each level of the stack holds $O(d \cdot b)$ bits, and this approach takes $\Theta(h \cdot d \cdot b)$ space overall. In contrast, the algorithm {\sc Add} \emph{adds} the values of the $d$ children to the existing $O(d \cdot b \log b)$ content of the $d$ registers, and uses polynomial interpolation to add the correct value of the node to the last register.

More prescisely, while {\sc Add} has no initial control over the content of $\hat{x}_1,\ldots,\hat{x}_d$, equation \eqref{eq:interp} allows {\sc Add} to perform a type of {\bf ``worst-case to \emph{arbitrary}-case''} reduction: in order to add the value of function $f_u$ on specific $v_{u1},\ldots,v_{ud}$ to the register $\hat{y}$, given that the initial space content is some \emph{arbitrary} $\hat{x}_1,\ldots,\hat{x}_d,\hat{y}$, it suffices to perform a running sum from $i = 1$ to $m$, where we multiply the $\hat{x}_1,\ldots,\hat{x}_d$ content by $\omega^i$, \emph{add} $v_{u1},\ldots,v_{ud}$ into the current space content, then evaluate the functions $\tilde{f}_{u,j}$ on that content, storing the result of the running sum into $\hat{y}$, which (after summing from $i=1,\ldots,m$) results in \emph{adding} the value $v_u$ into $\hat{y}$. That is, starting from any arbitrary values in the registers which are beyond our control, we can compute $f_u$ on any desired input registers.

The overall space usage of the above procedure is \[O(h \cdot \log (d \cdot b) + d \cdot b \log (d \cdot b)).\] The ``local'' storage is used to store current values $i \in [m]$, $r \in [d]$, and $O(1)$ bits to store a program counter (storing which of the two sets of recursive calls we're at), at each level of the tree. This takes $O(\log (d \cdot b))$ bits for each level of recursion, and $O(h \cdot \log (d \cdot b))$ space overall. The node index $u \in [d]^{\star}$ is also stored, which takes $O(h \cdot \log d)$ bits.

The ``global'' storage holds the content $(\hat{x}_1,\ldots,\hat{x}_d,\hat{y})$. Each $\hat{x}_i$ and $y$ consist of $b$ elements of $\F$, which take $O(b \log (d \cdot b))$ bits each. To compute the multilinear extension $\widetilde{f}_{u,j}$ on $(\omega^i \cdot \hat{x}_1+v_{u1},\ldots,\omega^i \cdot \hat{x}_d+v_{ud})$, we follow equation \eqref{eq:multilinear}: we use $O(b)$ bits of space to sum over all $a_1,\ldots,a_d \in \{0,1\}^b$, and we use $O(\log |\F|) \leq O(\log(d\cdot b))$ extra space to evaluate the expression $\chi_{a_1,\ldots,a_d}(\omega^i \cdot \hat{x}_1+v_{u1},\ldots,\omega^i \cdot \hat{x}_d+v_{ud})$ and multiply it with the value $f_{u,j}(a_1,\ldots,a_d) \in \{0,1\}$.

Finally, we describe how to improve the space usage to $O(h \log (d\cdot b) + d \cdot b)$. The idea is to use ``low-degree extensions'' of $f_u$, rather than multilinear extensions, and to group the $b$-bit output of $f_u$ into approximately $\lceil b/\log |\F| \rceil$ blocks, instead of $b$ blocks. In more detail, we modify the polynomials $\widetilde{f}_{u,j}$ over $\F$ so that they have $O(d \cdot (b/\log|\F|))$ variables instead of $O(d \cdot b)$ variables, and the $j$ ranges over a set $\{1,\ldots,\lceil cb/\log|\F|\rceil \}$ for a constant $c > 0$, instead of $[b] = \{1,\ldots,b\}$. To accommodate this, we will need to adjust the order of $\F$ slightly. 

Let $|\F| = 2^q$, for an even integer $q$ to be determined later (recall $\F$ is characteristic two). Let $S \subseteq \F$ be a set of cardinality $2^{q/2}$, which is put in one-to-one correspondence with the elements of $\{0,1\}^{q/2}$, and let $t = \lceil b/\log |S|\rceil$. We can then view the function $f_u : \{0,1\}^{d \cdot b} \rightarrow \{0,1\}^b$ as instead having domain $(S^t)^d$, and co-domain $S^t$. That is, we can think of the output of $f_u$ as the concatenation of $t$ subfunctions $(f_{u,1},\ldots,f_{u,t})$, with each $f_{u,j} : (S^t)^d \rightarrow S$. For each $j=1,\ldots,t$, the multilinear polynomial $\widetilde{f}_{u,j}$ of equation \eqref{eq:multilinear} can then be replaced by another polynomial in $d \cdot t$ variables over $\F$:
\begin{align}\label{eq:low-degree}
\widetilde{f}_{u,j}(\vec{x_1},\ldots,\vec{x_d}) = \sum_{a_1,\ldots,a_d \in S^t} \chi_{a_1,\ldots,a_d}(\vec{x_1},\ldots,\vec{x_d}) \cdot f_{u,j}(a_1,\ldots,a_d),
\end{align}
where again $\chi_{a_1,\ldots,a_d}(\vec{x_1},\ldots,\vec{x_d})$ is a polynomial over $\F$ such that $\chi_{a_1,\ldots,a_d}(a_1,\ldots,a_d) = 1$, and $\chi_{a_1,\ldots,a_d}$ vanishes on all $a'_1,\ldots,a'_d \in S^t$ such that $(a'_1,\ldots,a'_d) \neq (a_1,\ldots,a_d)$. Such a polynomial $\chi_{a_1,\ldots,a_d}$ can be constructed with degree $d\cdot(|S|-1)\cdot t = d \cdot (\sqrt{|\F|} - 1) \cdot t$. As long as this quantity is less than $|\F|$, we can pick an $m$-th root of unity $\omega$ with $m = |\F|-1$ and we may apply the polynomial interpolation formula of equation \eqref{eq:interp}. WLOG, we may assume $d$ and $b$ are even powers of two (otherwise, we can round them up to the closest such powers of two). Setting $2^q = d^2 b^2$, we have \[d \cdot (|S| - 1) \cdot t \leq d \cdot (d \cdot b - 1) \cdot b < d^2 b^2 = |\F|.\] 

As a result, each of the registers $\hat{x}_1,\ldots,\hat{x}_d,\hat{y}$ can now be represented with $t = \lceil b/\log |S| \rceil$ elements of $\F$, rather than $b$ elements of $\F$ as before. Since $\log|S| =\Theta(\log|\F|)$, each such register can now be represented with $O(d \cdot b)$ bits instead, and the new polynomials of \eqref{eq:low-degree} can still be evaluated in $O(b)$ space.
\end{document}